\documentclass[11pt,a4paper]{article}
\usepackage[margin=25mm]{geometry}
\usepackage{lmodern}
\usepackage[T1]{fontenc}
\usepackage{microtype,cite}
\usepackage{authblk}

\usepackage{amsmath,amssymb,amsthm,mathtools}
\usepackage{enumitem,booktabs,tabularx}
\usepackage{xcolor}
\usepackage[colorlinks=true,linkcolor=blue!45!black,citecolor=blue!45!black,urlcolor=blue!45!black]{hyperref}
\hypersetup{pdftitle={A Generalized Stein Lemma for Quantum Channels},
 pdfauthor={Minbo Gao, Zhengfeng Ji, Chenghua Liu},
 pdfsubject={Generalized Stein identity for quantum channels with correlated alternatives},
 pdfkeywords={Quantum channels, hypothesis testing, relative entropy, asymptotic equipartition, resource theories}}
\usepackage[capitalise,nameinlink,noabbrev]{cleveref}
\newtheorem{theorem}{Theorem}
\newtheorem{lemma}[theorem]{Lemma}
\newtheorem{proposition}[theorem]{Proposition}
\newtheorem{corollary}[theorem]{Corollary}
\theoremstyle{definition}
\newtheorem{definition}[theorem]{Definition}
\newtheorem{example}[theorem]{Example}
\theoremstyle{remark}
\newtheorem{remark}[theorem]{Remark}

\DeclareMathOperator{\Tr}{Tr}
\DeclareMathOperator{\Rea}{Re}
\DeclareMathOperator{\Ad}{Ad}
\DeclareMathOperator{\supp}{supp}
\DeclareMathOperator{\CPTP}{CPTP}

\newcommand{\id}{\mathrm{id}}
\newcommand{\Dens}{\mathrm D}
\newcommand{\F}{\mathfrak F}
\newcommand{\N}{\mathcal N}
\newcommand{\M}{\mathcal M}
\newcommand{\R}{\mathcal R}
\newcommand{\ch}{\mathrm{ch}}

\newcommand{\ket}[1]{|#1\rangle}
\newcommand{\bra}[1]{\langle#1|}
\newcommand{\opnorm}[1]{\lVert#1\rVert_\infty}
\newcommand{\trnorm}[1]{\lVert#1\rVert_1}
\newcommand{\dnorm}[1]{\lVert#1\rVert_\diamond}
\newcommand{\cpge}{\mathrel{\ge_{\mathrm{CP}}}}
\newcommand{\cple}{\mathrel{\le_{\mathrm{CP}}}}
\newcommand{\eps}{\varepsilon}
\newcommand{\loss}[1]{#1^{2/3}\log_2(#1+1)}
\setlist[enumerate]{leftmargin=*,itemsep=2pt,topsep=4pt}
\begin{document}
\title{A Generalized Stein Lemma for Quantum Channels}
\author[1,2]{Minbo Gao\thanks{\href{mailto:gmb17@tsinghua.org.cn}{\texttt{gmb17@tsinghua.org.cn}}}}
\author[3]{Zhengfeng Ji\thanks{\href{mailto:jizhengfeng@tsinghua.edu.cn}{\texttt{jizhengfeng@tsinghua.edu.cn}}}}
\author[1,2]{Chenghua Liu\thanks{\href{mailto:liuch.russell@gmail.com}{\texttt{liuch.russell@gmail.com}}}}
\affil[1]{Institute of Software, Chinese Academy of Sciences, Beijing, China}
\affil[2]{University of Chinese Academy of Sciences, Beijing, China}
\affil[3]{Department of Computer Science and Technology, Tsinghua University, Beijing, China}
\date{}

\maketitle
\vspace{-2.5em}

\begin{abstract}
  We prove a generalized quantum Stein lemma for parallel discrimination of an
  arbitrary finite-dimensional channel against families of alternative channels.
  The alternatives are compact and convex, closed under tensor products, and
  contain a faithful replacer.
  At every fixed type-I error tolerance, the optimal type-II error exponent
  equals the regularized Umegaki channel relative entropy minimized over the
  alternatives, and both defining limits exist.
  Inputs may be entangled across channel uses and with a reference system.
  The main technical result constructs exactly trace-preserving approximations
  with exponentially small diamond error, dominated in completely positive order
  by free channels at every rate above the common exponent.
  The proof combines uniform auxiliary map approximation, iterative reduction
  of the domination rate, and tensor amplification.
  It yields asymptotic equipartition with trace-preserving smoothing at every
  subexponentially vanishing error, an exponential strong converse, and
  stability under vanishing diamond-norm perturbations.
  Under additional permutation invariance and closure under insertion of a fixed
  faithful input state and discarding of its output, we also prove an explicit
  finite-block completion bound.
  This quantitative construction uses weighted discarding, local operator
  corrections, and comparison with auxiliary extensions.
  Applications include state-preserving, coherence-restricted, covariant,
  entanglement-breaking, and positive-partial-transpose channels.
  For isometric targets against entanglement-breaking or
  positive-partial-transpose alternatives, we obtain exact finite-block testing
  and smoothing formulas.
\end{abstract}

\section{Introduction}\label{sec:introduction}
Quantum hypothesis testing gives an operational meaning to relative
entropy. Consider two states $\rho$ and $\sigma$, and let the effect
$0\le Q_n\le I$ represent acceptance of $\rho^{\otimes n}$.
The two error probabilities are
\[
 \alpha_n=1-\Tr Q_n\rho^{\otimes n},\qquad
 \beta_n=\Tr Q_n\sigma^{\otimes n}.
\]
In the asymmetric problem, one bounds $\alpha_n$ by a fixed
$\varepsilon\in(0,1)$ and minimizes $\beta_n$.
The quantum Stein lemma identifies the limiting exponent of this
minimum with the Umegaki relative entropy
$D(\rho\Vert\sigma)=\Tr\rho(\log_2\rho-\log_2\sigma)$,
with value $+\infty$ when the support of $\rho$ is not contained
in that of $\sigma$
\cite{HiaiPetz,OgawaNagaoka}. Thus relative entropy measures the
asymptotic rate at which an alternative can be excluded while the
target is accepted with a prescribed probability.

For a quantum process, the experimenter must choose an input as well
as a measurement. An input may be entangled with a reference that
does not pass through the process; with several channel uses, it may
also be entangled across those uses. A further difficulty arises when
the alternative is a family of processes rather than a single known
channel. Such families occur naturally in resource theories, where
membership in a specified set expresses the absence of the resource
under consideration. Detecting a resource then requires one
experiment that rejects every member of the alternative family.

This paper treats parallel access to a target channel
$\N:\mathrm L(A)\to\mathrm L(B)$. At blocklength $n$, the target is $\N^{\otimes n}$,
whereas an alternative is an arbitrary member of
$\F_n\subseteq\CPTP(A^{\otimes n}\to B^{\otimes n})$.
Here $A$ and $B$ are finite-dimensional input and output Hilbert
spaces, and CPTP denotes completely positive trace-preserving maps.
We write $\F=(\F_n)_{n\ge1}$ and call its members free channels.
The alternatives may correlate all inputs and outputs within the
block. The tester prepares one reference-assisted input and measures
the joint reference and output after all channel uses. There is no
restriction to classical inputs or to product input states. The access
model is nevertheless parallel: it does not insert intermediate
operations between uses of the unknown process.

\subsection{Results and scope}
Let $\beta_{\varepsilon,n}$ be the smallest worst-case alternative
acceptance probability among testers whose target acceptance is at
least $1-\varepsilon$. Write $D_{\ch}$ for Umegaki relative entropy
optimized over a reference-assisted input shared by the target and
the specified alternative. Our main result is
\begin{equation}
 \lim_{n\to\infty}-\frac1n\log_2\beta_{\varepsilon,n}
 =\lim_{n\to\infty}\frac1n
   \inf_{\M\in\F_n}D_{\ch}(\N^{\otimes n}\Vert\M).
 \label{eq:intro-identity}
\end{equation}
Both limits exist, and their value is independent of
$\varepsilon\in(0,1)$.
The three assumptions in \cref{def:axioms} are compactness and convexity
of the alternative sets, closure under tensor products, and membership
of a faithful replacer. These assumptions suffice for all the
asymptotic results below.

The main technical result gives exponential approximation at any
domination rate $S$ strictly above the common rate in
\eqref{eq:intro-identity}. We write $\Phi\cpge\Psi$ when
$\Phi-\Psi$ is completely positive, and $\|\cdot\|_\diamond$ for the
diamond norm, defined in \cref{sec:setting}.
For all sufficiently large $n$, there are a free channel
$\mathcal S_n\in\F_n$ and a completely positive trace-preserving map
$\mathcal L_n$ such that
\begin{equation*}
 \mathcal L_n\cple2^{nS}\mathcal S_n,\qquad
 \dnorm{\mathcal L_n-\N^{\otimes n}}\le K_Se^{-\gamma_S n},
\end{equation*}
for constants $K_S,\gamma_S>0$.
The same approximation and dominator work for every input, reference,
and acceptance effect. This uniform control yields an exponential
strong converse: above the Stein rate, the target acceptance probability
is at most $K2^{-cn}$ for suitable $K,c>0$.
It also gives asymptotic equipartition (AEP) for trace-preserving
smoothing with $0<\delta_n\le\bar\delta<2$ and
$\log_2(1/\delta_n)=o(n)$, including every fixed diamond radius
in $(0,2)$. The Stein exponent is stable under vanishing diamond
perturbations of the target sequence.

A second contribution is an explicit finite-block completion bound.
For this bound we additionally assume permutation invariance and
marginal closure: inserting one fixed faithful input state and
discarding the corresponding output preserves membership in the
family. These conditions are stated in \cref{def:quantitative-axioms}.
For every fixed $\varepsilon\in(0,1)$ and $0<\delta\le1/16$,
there are a free channel $\mathcal S_n$ and a CPTP map $\mathcal L_n$
with
\begin{equation*}
 \begin{aligned}
 \mathcal S_n&\cpge
 \beta_{\varepsilon,n}2^{-g_\varepsilon(n,\delta)}\mathcal L_n,
 &\dnorm{\mathcal L_n-\N^{\otimes n}}&\le\delta,\\
 g_\varepsilon(n,\delta)&=K_\varepsilon n^{2/3}
                   \log_2\frac{n+1}{\delta}.
 \end{aligned}
\end{equation*}
The constant $K_\varepsilon$ depends only on the local dimensions,
the minimum eigenvalues of the witnessing states, and $\varepsilon$.
Thus the bound controls the domination cost at a specified blocklength
and accuracy relative to the finite-block testing coefficient.
The overhead is sublinear for fixed or inverse-polynomial accuracy.

All five conditions hold for the resource families in
\cref{tab:families}, including state-preserving, maximally incoherent,
dephasing-covariant, symmetry-covariant, entanglement-breaking, and
positive-partial-transpose channels.
For isometric targets against the last two families, we compute
testing and smoothing exactly at every blocklength.

\subsection{Relation to earlier work}
The generalized state Stein problem was formulated by Brand\~ao and
Plenio \cite{BrandaoPlenio}. Following the identification of a gap in
the original proof \cite{BertaEtAl}, proofs were obtained by Hayashi
and Yamasaki \cite{HayashiYamasakiState} and Lami \cite{Lami}.
Families of state hypotheses built from mixtures of tensor powers were
studied by Berta, Brand\~ao, and Hirche \cite{BertaBrandaoHirche}.
Related developments include generalized asymptotic equipartition
\cite{FangFawziFawzi} and error exponents for families of correlated
state hypotheses \cite{FangHayashi}. With a one-dimensional input,
our asymptotic theorem reduces to the generalized state problem without
permutation or partial-trace closure, as already established for states
by Hayashi and Yamasaki \cite{HayashiYamasakiState}.

For discrimination between two quantum channels, Gao, Ji, and Liu
\cite{GaoJiLiu} established the channel Stein lemma and an exponential
strong converse under arbitrary adaptive strategies.
Bergh, Datta, and Salzmann characterized the parallel Stein exponent
for composite hypotheses generated by independent, not necessarily
identically distributed channel uses and their convex hulls, with
$\varepsilon\downarrow0$ taken after the blocklength $\liminf$ or $\limsup$
\cite[Thm.~10]{BerghDattaSalzmann}. Our alternatives include correlated
block channels subject to \cref{def:axioms}. Generalized Stein lemmas for
classical--quantum channels were established by Hayashi and Yamasaki
\cite{HayashiYamasakiCQ} and, independently, by Bergh, Datta, and
Khaitan \cite{BerghDattaKhaitan}. A classical--quantum channel
measures in a fixed input basis and prepares an output state
depending on the result. The present theorem allows general quantum
inputs and outputs, under the structural assumptions above.
In particular, the marginal axiom used in our finite-block completion
bound is not required for our asymptotic theorem and should not be read
as a hypothesis shared by all of the classical--quantum results.

Exact trace preservation creates a separate issue from the testing
problem. Gour \cite{Gour} gives channel pairs for which max-relative
entropy smoothed over trace-preserving maps does not converge to
the regularized relative entropy when the dominator is prescribed
to be the tensor power of one fixed channel. Our smoothing instead
optimizes the dominator over $\F_n$. Convexity, tensor closure, and the free replacer allow the
dominator to change during amplification and completion. A singleton
family $\F_n=\{\M^{\otimes n}\}$ satisfies the faithful-replacer
assumption only if $\M$ itself is a faithful replacer. Thus the present
theorem does not establish a general
fixed-dominator channel equipartition theorem.

\subsection{Structure of the argument}
The proof adapts the uniform Stinespring approximation and tensor
amplification methods of \cite{GaoJiLiu} to families of alternative channels.
Testing minimax supplies an alternative that is hardest to distinguish
from the target at a fixed error tolerance.
A semidefinite representation of the testing quantity then produces
a uniform auxiliary map approximation for prescribed Stinespring dilations.
Keeping their auxiliary spaces fixed allows two approximations relative
to a common free dominator to be subtracted and tensorized.

The amplification argument has two stages.
First, a sufficiently accurate fixed block can be repeated and its
tensor expansion truncated to give exponential accuracy, at an
arbitrarily small increase in domination rate.
Second, a rough approximation from testing is combined with an
exponentially accurate approximation at a higher rate. This produces
accurate blocks at an improved rate.
Starting from exact domination by the faithful replacer, finitely
many improvements reach every rate above the lower limiting testing
exponent.

The resulting approximations compare the upper and lower limiting
testing exponents at every pair of fixed error tolerances.
This proves ordinary convergence and error independence.
A direct entropy comparison and the weak converse identify the
common limit with the regularized Umegaki channel relative entropy.
The proof establishes these limits without an additivity assumption.

The quantitative finite-block bound uses a separate construction.
A positive-overlap branch is corrected by local operator expansions
obtained from permutation averaging and weighted discarding.
The additional family conditions convert these corrections into
free domination, and a replacer restores trace preservation and the
discarded sites. Comparison with auxiliary extensions reduces a
general target to the isometric construction.

\Cref{sec:setting} states the model and main results.
\Cref{sec:one-shot} develops testing minimax, auxiliary map
approximation, and trace-defect completion.
\Cref{sec:exponential} proves exponential approximation and the
Stein identity, and \cref{sec:consequences} derives their asymptotic
consequences.
\Cref{sec:quantitative} states the additional assumptions and the
quantitative finite-block results.
Resource families and exact benchmarks appear in
\cref{sec:applications}, followed by the discussion in
\cref{sec:discussion}.
\Cref{app:quantitative-proof} proves quantitative completion;
\cref{app:alternative-proof} gives symmetric entropy estimates and
an alternative limit proof under all five conditions.

\section{Model and main results}\label{sec:setting}
\subsection{Operator and channel conventions}
All Hilbert spaces considered in this paper are finite dimensional.
We write $\mathrm L(H,K)$ for the linear operators from $H$ to $K$, and
$\mathrm L(H)=\mathrm L(H,H)$.
The adjoint of an operator or map is denoted by a dagger; for maps it is taken
in the Hilbert--Schmidt inner product.
We denote the identity operator by $I_H$ and the identity map by $\id_H$,
omitting subscripts when the space is clear.
For a Hermitian operator $X$, $\lambda_{\min}(X)$ is its smallest eigenvalue,
and $\supp X$ is its range support.
The norm of a vector is its Hilbert-space norm.
Partial traces are denoted by $\Tr_K$; for a state $\rho$ on $A\otimes B$, write
$\rho_B=\Tr_A\rho$.
All transposes are taken in fixed orthonormal bases, using product bases on
tensor powers.
Our operator and channel conventions follow \cite{Watrous}, with the reference
factor written first in Choi operators.
Write $\Dens(H)$ for the density operators on $H$ and $\CPTP(A\to B)$ for the
CPTP maps from $\mathrm L(A)$ to $\mathrm L(B)$.
Pure states are represented by unit vectors $\ket\psi$, with density operator
$\ket\psi\bra\psi$.
Mixed states and reduced states are written as density operators, typically
denoted by $\rho$, $\sigma$, $\tau$, or $\omega$.
A state is called faithful if it is positive definite, equivalently full rank.
In finite dimensions its smallest eigenvalue is then strictly positive.
The operator and trace norms are denoted by $\|\cdot\|_\infty$ and
$\|\cdot\|_1$, respectively.
We use $\log_2$ for information quantities and $\ln$ for the natural logarithm.
For a linear operator $K$, let $\Ad_K(X)=KXK^{\dagger}$.
In asymptotic estimates, subscripts on $O(\cdot)$ indicate the parameters on
which the implicit constant may depend; $o(n)$ denotes a quantity whose ratio to
$n$ tends to zero.
For any square operator $X$, its Hermitian part is $\Rea X=(X+X^{\dagger})/2$;
an inequality involving $\Rea X$ is understood in the Loewner order.
For Hermiticity-preserving maps, $\Phi\cpge\Psi$ means that $\Phi-\Psi$ is
completely positive (CP).
Pre- and postcomposition with a CP map, and tensoring with a CP map, preserve
this order, as do positive linear combinations.
In particular, if $\Phi_i\cpge\Psi_i\cpge0$ for $i=1,2$, then
$\Phi_1\otimes\Phi_2\cpge\Psi_1\otimes\Psi_2$.

For states $\rho,\sigma$ on the same space, the relative entropy is
\begin{equation*}
 D(\rho\Vert\sigma)=
 \Tr\rho(\log_2\rho-\log_2\sigma),
\end{equation*}
with value $+\infty$ unless $\supp\rho\subseteq\supp\sigma$.
Functions of positive operators are defined spectrally, with
$0\log_2 0=0$ in entropy expressions.
When $\sigma$ is singular and $\supp\rho\subseteq\supp\sigma$,
the term $\Tr\rho\log_2\sigma$ is evaluated on $\supp\sigma$.
We use the data-processing and joint-convexity properties of relative entropy.
For a linear map $\Phi:\mathrm L(H)\to\mathrm L(K)$, the diamond
norm is the completely bounded trace norm,
\[
 \dnorm\Phi=\sup_{R}\sup_{0\ne Z\in\mathrm L(R\otimes H)}
       \frac{\trnorm{(\id_R\otimes\Phi)(Z)}}{\trnorm Z}.
\]
A reference of dimension $\dim H$ suffices. The diamond norm
bounds the trace-norm difference of outputs for every reference-assisted
input, is nonincreasing under pre- and postcomposition by channels,
and equals one on channels. Consequently a diamond distance at
most $\delta$ changes any tester's acceptance probability by at
most $\delta$. We use this bound throughout; for two channels the
sharper bound $\delta/2$ also holds because their output difference
has trace zero.
For a CP map $\Phi$, one has
$\dnorm\Phi=\opnorm{\Phi^{\dagger}(I)}$. The diamond norm is
multiplicative under tensor products; in particular, tensoring
with a channel preserves it.
For operators $K$ and $L$ with the same domain and codomain,
\begin{equation}
 \dnorm{\Ad_K-\Ad_L}
 \le(\opnorm K+\opnorm L)\opnorm{K-L}.
 \label{eq:single-kraus-continuity}
\end{equation}
To see this, let $\widetilde K=I_R\otimes K$ and
$\widetilde L=I_R\otimes L$. For an arbitrary, not necessarily
positive, operator $Z$,
\begin{align*}
 \widetilde K Z\widetilde K^{\dagger}-\widetilde L Z\widetilde L^{\dagger}
 &=(\widetilde K-\widetilde L)Z\widetilde K^{\dagger}
       +\widetilde L Z(\widetilde K-\widetilde L)^{\dagger},\\
 \trnorm{\widetilde K Z\widetilde K^{\dagger}-\widetilde L Z\widetilde L^{\dagger}}
 &\le(\opnorm K+\opnorm L)\opnorm{K-L}\trnorm Z.
\end{align*}
Here we used $\|UZV\|_1\le\|U\|_\infty\|Z\|_1\|V\|_\infty$
and invariance of operator norm under tensoring with an identity.
Taking the supremum proves \eqref{eq:single-kraus-continuity}.
This estimate will convert an operator approximation of a Kraus
branch into a diamond approximation during trace completion.

We use the unnormalized Choi operator
\begin{equation}
 C_\Phi=(\id_H\otimes\Phi)(\ket{\Gamma_H}\bra{\Gamma_H}),
 \qquad \ket{\Gamma_H}=\sum_{j=1}^{\dim H}\ket j\ket j.
 \label{eq:choi}
\end{equation}
For an operator $L:H\to K$, we write $\ket L=(I_H\otimes L)\ket{\Gamma_H}$ for
its vectorization.
Thus $\ket L$ is a vector in $H\otimes K$, whereas $L\ket j$ is a vector in $K$.
One has $C_{\Ad_L}=\ket L\bra L$ and $\langle L\mid M\rangle=\Tr L^{\dagger}M$.
The Choi criterion gives $\Phi\cpge\Psi$ if and only if $C_\Phi\ge C_\Psi$.
A map $\Phi:\mathrm L(H)\to\mathrm L(K)$ is CPTP if and only if $C_\Phi\ge0$ and $\Tr_K C_\Phi=I_H$.

\subsection{Alternative families}
Fix $A,B$, put $a=\dim A$, $b=\dim B$, and fix a channel
$\N:\mathrm L(A)\to\mathrm L(B)$. Abbreviate $\N_n=\N^{\otimes n}$ and
$A^n=A^{\otimes n}$, with analogous notation for other systems.
The symbol $A_i$ denotes the $i$th tensor factor of $A^n$,
and $[n]=\{1,\ldots,n\}$.
For a state $\omega\in\Dens(B)$, the replacer channel is
$\R_\omega(X)=\Tr(X)\omega$.

A channel is \emph{free} when it belongs to the family specified
below. The axioms relate these sets across blocklengths and allow
each member of $\F_n$ to correlate all $n$ uses.

\begin{definition}[Admissible free families]\label{def:axioms}
A sequence $\F=(\F_n)_{n\ge1}$, with
$\F_n\subseteq\CPTP(A^n\to B^n)$, is admissible if:
\begin{enumerate}[label=\textup{(F\arabic*)},ref=F\arabic*]
\item\label{ax:convex} Each $\F_n$ is nonempty, compact, and convex.
\item\label{ax:tensor} $\M\otimes\mathcal L\in\F_{n+m}$ whenever
$\M\in\F_n$ and $\mathcal L\in\F_m$.
\item\label{ax:replacer} There is a fixed positive-definite
$\omega\in\Dens(B)$ with $\R_\omega\in\F_1$.
\end{enumerate}
Tensor closure implies $\R_\omega^{\otimes n}\in\F_n$.
\end{definition}

The word free refers only to membership in these sets, and admissible
means that the sequence satisfies these three conditions. In
particular, no closure under arbitrary input or output processing
is assumed. The state $\omega$ appearing below is always the witness
in \ref{ax:replacer}; write $w=\lambda_{\min}(\omega)>0$.
Different witnessing states can change the constants in our estimates,
but the testing problem itself is determined by $\N$ and $\F$.
Unless stated otherwise, $\F$ is assumed admissible throughout.
The quantitative results in \cref{sec:quantitative} impose two
additional conditions, stated there as \ref{ax:permutation} and
\ref{ax:marginal}.

\subsection{Testing and relative entropy}

We may take the reference-assisted input to be a pure state
$\ket\psi\in R\otimes A^n$. Indeed, any mixed input can be
purified by enlarging the reference. Extending the acceptance
effect by the identity preserves all testing probabilities, and
data processing shows that purification cannot decrease the
output relative entropy. Schmidt compression then allows
$R\simeq A^n$. All input state vectors in the optimizations below have
unit norm. Write $\varrho_R=\Tr_{A^n}\ket\psi\bra\psi$ for the
reference marginal.

For a unit vector $\ket\psi\in R\otimes A^n$, define the output
density operators
\begin{align*}
 \rho_n^\psi&=(\id_R\otimes\N_n)(\ket\psi\bra\psi),\\
 \sigma_\M^\psi&=(\id_R\otimes\M)(\ket\psi\bra\psi).
\end{align*}
The target output is denoted by $\rho$ and an alternative output by
$\sigma$; superscripts and subscripts specify the input and channel
when needed. The stabilized channel relative entropy and its distance
to the free family are
\begin{align*}
 D_{\ch}(\N_n\Vert\M)
 &=\sup_{R,\ket\psi}D(\rho_n^\psi\Vert\sigma_\M^\psi),
 \\
 E_n&=\inf_{\M\in\F_n}D_{\ch}(\N_n\Vert\M).
\end{align*}
The supremum is over normalized input vectors; as explained above,
$R\simeq A^n$ suffices. In $D_{\ch}(\N_n\Vert\M)$ the optimizing
input may depend on $\M$. The testing problem below instead chooses
one input and one measurement for the entire alternative family.
For $0<\eps<1$, define
\begin{equation}
 \beta_{\eps,n}
 =\inf_{\substack{R,\ket\psi,\ 0\le Q\le I\\
                  \Tr Q\rho_n^\psi\ge1-\eps}}
       \sup_{\M\in\F_n}\Tr Q\sigma_\M^\psi.
 \label{eq:beta}
\end{equation}
A tester is the pair consisting of the input state $\ket\psi$ and the
acceptance effect $Q$, together with its reference space $R$.
The two measurement effects are $Q$ and $I-Q$: the outcome associated
with $Q$ declares the channel to be the target, and the other outcome
rejects the target.
Its type-I and type-II errors are, respectively,
\begin{align*}
 \alpha(\ket\psi,Q)&=1-\Tr Q\rho_n^\psi,\\
 \beta(\ket\psi,Q;\M)&=\Tr Q\sigma_\M^\psi.
\end{align*}
Feasibility means $\langle\psi|\psi\rangle=1$, $0\le Q\le I$, and
$\alpha(\ket\psi,Q)\le\varepsilon$.
Thus \eqref{eq:beta} minimizes the largest acceptance probability over
free alternatives, subject to accepting $\N_n$ with probability at least
$1-\varepsilon$. The target channel determines which testers satisfy
this constraint through the output state $\rho_n^\psi$.
For a simple feasible tester, choose any normalized input and set
$Q=(1-\varepsilon)I$. Operationally, this tester ignores the output and
declares the channel to be the target with probability $1-\varepsilon$.
Since every output state has trace one,
\[
 \Tr Q\rho_n^\psi=\Tr Q\sigma_\M^\psi=1-\varepsilon
 \qquad(\M\in\F_n).
\]
Its type-I error is therefore $\varepsilon$, and its worst-case type-II
error is $1-\varepsilon$, proving the elementary bound
$\beta_{\eps,n}\le1-\varepsilon$.
There is one input state $\ket\psi$ and one effect $Q$ for all alternatives; we
call such a tester \emph{uniform}, meaning independent of the alternative.
The purification argument above ensures that these optimizations cover all
physically admissible input states and effects.

To relate this composite testing problem to testing against a single
alternative, we introduce the corresponding single-alternative optimum.
For a target channel $\mathcal T:\mathrm L(H)\to\mathrm L(K)$ and a single
alternative channel $\M:\mathrm L(H)\to\mathrm L(K)$, write
$b_\eps(\mathcal T,\M)$ for the optimal reference-assisted parallel type-II
error at target channel $\mathcal T$ and type-I error at most $\eps$.
The infimum below is over finite-dimensional reference systems $R$, unit vectors
$\ket\psi\in R\otimes H$, and acceptance effects $Q$ on $R\otimes K$:
\[
 b_\eps(\mathcal T,\M)
 =\inf_{\substack{R,\ket\psi,\ 0\le Q\le I\\
       \Tr Q(\id_R\otimes\mathcal T)(\ket\psi\bra\psi)\ge1-\eps}}
       \Tr Q(\id_R\otimes\M)(\ket\psi\bra\psi).
\]
In $b_\eps(\mathcal T,\M)$, the tester may depend on the specified alternative
$\M$, whereas $\beta_{\eps,n}$ requires one tester for all alternatives in
$\F_n$.
The minimax theorem proved in \cref{lem:testing-minimax} establishes the
relation
\[
 \beta_{\eps,n}=\max_{\M\in\F_n}b_\eps(\N_n,\M).
\]

For the asymptotic analysis, we write $a_{\eps,n}=-\log_2\beta_{\eps,n}$.
The faithful replacer ensures $\beta_{\eps,n}>0$; the quantitative bounds are
given in \cref{lem:bounds}, with the constant
$C=\log_2 b-\log_2\lambda_{\min}(\omega)$.

\subsection{Main theorems}
\begin{theorem}[Generalized Stein identity for quantum channels]\label{thm:main}
Let $\N:\mathrm L(A)\to\mathrm L(B)$ be arbitrary and let $\F$ satisfy
\ref{ax:convex}, \ref{ax:tensor}, and \ref{ax:replacer}.
Then there is a finite $R_\F(\N)\ge0$ such that
\begin{equation*}
 \lim_{n\to\infty}-\frac1n\log_2\beta_{\eps,n}
 =R_\F(\N)
 =\lim_{n\to\infty}\frac{E_n}{n}
 \qquad(0<\eps<1).
\end{equation*}
Both limits are ordinary limits over all positive integers.
Writing $w=\lambda_{\min}(\omega)$, one has
$R_\F(\N)\le \log_2 b-\log_2 w$.
\end{theorem}
The rate need not be strictly positive when $\N\notin\F_1$.
\Cref{sec:nonfaithfulness} gives an admissible family whose correlated
alternatives contain a fixed target component at every blocklength,
forcing $R_\F(\N)=0$.

The main technical result constructs exactly trace-preserving
approximations with exponential accuracy at every rate above the
Stein rate. Their domination by a free channel controls all
reference-assisted tests simultaneously.

\begin{theorem}[Exponential free-channel approximation]
\label{thm:exponential-completion}
Under the assumptions of \cref{thm:main}, for every $S>R_\F(\N)$
there are $K_S,\gamma_S>0$ such that, for all sufficiently large $n$,
one can find $\mathcal S_n\in\F_n$ and
$\mathcal L_n\in\CPTP(A^n\to B^n)$ satisfying
\begin{equation}
 \begin{aligned}
  \mathcal L_n&\cple2^{nS}\mathcal S_n,\\
  \dnorm{\mathcal L_n-\N_n}&\le K_Se^{-\gamma_S n}.
 \end{aligned}
 \label{eq:exponential-completion}
\end{equation}
The dominator is chosen independently of the input, reference, and test.
The constants and the blocklength threshold may depend on $\N$, $\F$,
and $S$.
\end{theorem}

\Cref{sec:one-shot} develops the testing and approximation tools;
\cref{sec:exponential} uses them to prove
\cref{thm:main,thm:exponential-completion} jointly.
The proof starts from a lower limiting testing exponent and establishes
its convergence before identifying the entropy rate.
Under two additional conditions, \cref{thm:completion} gives an
explicit finite-block domination bound in terms of $\beta_{\eps,n}$.
Its constant is uniform over targets and families with the specified
dimensions and witnessing-state eigenvalue bounds.

\section{Testing and uniform channel approximation}\label{sec:one-shot}

In this section, we provide the four ingredients of the main proof.
First, we apply replacer domination to obtain uniform bounds and a weak
converse (\cref{lem:bounds}).
Second, testing minimax allows us to select the hardest alternative
(\cref{lem:testing-minimax}).
Third, auxiliary map approximation converts its testing bounds into a single
operator approximation valid for every input
(\cref{lem:auxiliary-map-approximation}).
Finally, a completion step restores exact trace preservation at a cost
controlled by the trace defect (\cref{lem:defect-completion}).
For all family-dependent statements, we rely only on \cref{def:axioms}.

\subsection{Uniform bounds and testing minimax}
\begin{lemma}\label{lem:bounds}
Assume only \ref{ax:convex}, \ref{ax:tensor}, and \ref{ax:replacer}.
Set $C=\log_2 b-\log_2 w$, where
$w=\lambda_{\min}(\omega)>0$. For every $n$ and $0<\eps<1$,
\begin{align}
 0\le E_n&\le nC,\label{eq:En-bound}\\
 (1-\eps)2^{-nC}\le\beta_{\eps,n}&\le1-\eps,\label{eq:beta-bound}\\
 (1-\eps)(-\log_2\beta_{\eps,n})&\le E_n+1.
 \label{eq:weak-converse}
\end{align}
\end{lemma}
\begin{proof}
We first prove the dimension bound
\begin{equation}
 \rho_{RB}\le b\,\rho_R\otimes I_B,
 \qquad \rho_{RB}\in\Dens(R\otimes B).
 \label{eq:dimension-order}
\end{equation}
For a pure state, write a Schmidt decomposition
$\ket\phi=\sum_{j=1}^r\sqrt{p_j}\ket{u_j}\ket{v_j}$, where
$r\le b$. For any vector $\ket x\in R\otimes B$,
Cauchy--Schwarz gives
\begin{align*}
 |\langle x|\phi\rangle|^2
 &\le r\sum_{j=1}^r p_j
          |\langle x|u_j, v_j\rangle|^2\\
 &\le b\,\bra x(\rho_R\otimes I_B)\ket x.
\end{align*}
In the last step, each $\ket{v_j}\bra{v_j}$ is bounded by $I_B$.
For a mixed state, decompose it into pure states, apply this bound
to each component, and sum with the mixture weights. The reference
marginals sum to $\rho_R$, proving \eqref{eq:dimension-order}.

We now apply this inequality to the channel outputs.
For an arbitrary unit input vector $\ket\psi\in R\otimes A^n$, set
\[
 \begin{aligned}
 \rho&=(\id_R\otimes\N_n)(\ket\psi\bra\psi),\\
 \sigma&=(\id_R\otimes\R_\omega^{\otimes n})(\ket\psi\bra\psi)\\
       &=\varrho_R\otimes\omega^{\otimes n}.
 \end{aligned}
\]
Because the joint map is $\id_R\otimes\N_n$ and $\N_n$ is trace preserving, the
reference marginal is unchanged.
That is, $\rho_R=\varrho_R$.
Apply \eqref{eq:dimension-order} with the whole output $B^n$, whose dimension is
$b^n$, and then use $\omega^{\otimes n}\ge w^nI_{B^n}$.
Equation~\eqref{eq:dimension-order} yields
\begin{equation}
 \rho\le b^nw^{-n}\sigma=2^{nC}\sigma. \label{eq:replacer-domination}
\end{equation}
The inequality in \eqref{eq:replacer-domination} implies
$\supp\rho\subseteq\supp\sigma$.
Restrict to $\supp\sigma$, where $\sigma$ is positive definite, and let $I$
denote the identity on this subspace.
For every $t>0$, $\rho+tI\le2^{nC}\sigma+tI$, so operator monotonicity of the
logarithm gives
\[
 \log_2(\rho+tI)\le\log_2(2^{nC}\sigma+tI).
\]
Multiplying by $\rho$ and taking the trace preserves this inequality.
Letting $t\downarrow0$, using $x\log_2x\to0$ as $x\downarrow0$, yields
\[
 \Tr\rho\log_2\rho
 \le\Tr\rho\log_2(2^{nC}\sigma)
 =nC+\Tr\rho\log_2\sigma,
\]
where we used $\Tr\rho=1$. Rearranging proves
$D(\rho\Vert\sigma)\le nC$.
This bound holds for every reference-assisted input, hence
$D_{\ch}(\N_n\Vert\R_\omega^{\otimes n})\le nC$.
Since $\R_\omega^{\otimes n}\in\F_n$, minimizing over the free
alternatives gives $E_n\le nC$. Nonnegativity of relative entropy
gives $E_n\ge0$, proving \eqref{eq:En-bound}.

For the testing lower bound, let $Q$ be any feasible acceptance
effect for the chosen input, so that $\Tr Q\rho\ge1-\eps$.
Equation~\eqref{eq:replacer-domination} implies
\[
 \Tr Q\sigma\ge2^{-nC}\Tr Q\rho
 \ge(1-\eps)2^{-nC}.
\]
The worst-case acceptance probability over $\F_n$ is at least
the acceptance probability of this free replacer. Taking the
infimum over feasible testers therefore gives
$\beta_{\eps,n}\ge(1-\eps)2^{-nC}$.
For the upper bound, choose any unit input vector and
$Q=(1-\eps)I$. Every channel output has trace one, so this
effect accepts both the target and every alternative with
probability $1-\eps$. It is thus feasible and gives
$\beta_{\eps,n}\le1-\eps$, completing \eqref{eq:beta-bound}.

For any feasible tester with input state $\ket\psi$ and effect $Q$, set
$\rho=(\id_R\otimes\N_n)(\ket\psi\bra\psi)$ and
$\sigma_\M=(\id_R\otimes\M)(\ket\psi\bra\psi)$.
Write $p=\Tr Q \rho\ge1-\eps$ and
$t=\sup_{\M\in\F_n}\Tr Q \sigma_\M$.
Let $h_2(p)=-p\log_2p-(1-p)\log_2(1-p)$ denote the binary
entropy, which is at most one.
For each $\M\in\F_n$, put $q=\Tr Q\sigma_\M$.
The binary measurement $\{Q,I-Q\}$ sends $\rho$ and $\sigma_\M$
to the probability distributions $(p,1-p)$ and $(q,1-q)$.
The definition of channel relative entropy and data processing give
\begin{align*}
 D_{\ch}(\N_n\Vert\M)
 &\ge D(\rho\Vert \sigma_\M)\\
 &\ge D((p,1-p)\Vert(q,1-q))\\
 &=p\log_2\frac1q+(1-p)\log_2\frac1{1-q}-h_2(p)\\
 &\ge p\log_2\frac1q-h_2(p).
\end{align*}
The last step drops the nonnegative term
$(1-p)\log_2(1/(1-q))$.
At boundary probabilities these expressions are interpreted by continuity;
if $q=0$, then $p>0$ makes the relative entropy infinite and the
desired lower bound is immediate.
Next, $q\le t\le1$ and $h_2(p)\le1$ imply
\begin{align*}
 p\log_2\frac1q-h_2(p)
 &\ge p\log_2\frac1t-1\\
 &\ge(1-\eps)\log_2\frac1t-1.
\end{align*}

For the fixed feasible tester, $t$ is its worst-case acceptance
probability over the entire family $\F_n$. Thus the lower bound
$(1-\eps)\log_2(1/t)-1$ is independent of the individual
alternative $\M$. Taking the infimum over $\M\in\F_n$ gives
\[
 E_n=\inf_{\M\in\F_n}D_{\ch}(\N_n\Vert\M)
 \ge(1-\eps)\log_2\frac1t-1.
\]
By the definition of $\beta_{\eps,n}$, there is a sequence of
feasible testers whose worst-case acceptance probabilities $t_j$
converge to $\beta_{\eps,n}$. Applying the preceding bound to
each tester and taking $j\to\infty$ yields
\[
 E_n\ge(1-\eps)\log_2\frac1{\beta_{\eps,n}}-1,
\]
which is \eqref{eq:weak-converse}. The lower bound
$\beta_{\eps,n}\ge(1-\eps)2^{-nC}>0$ ensures that the logarithm
is finite and continuous at this limit.
\end{proof}

The tester representation makes the worst-case testing optimization a
compact convex minimax problem.

\begin{lemma}\label{lem:testing-minimax}
For a channel $\mathcal T:\mathrm L(H)\to\mathrm L(K)$, a nonempty compact convex
set $\mathcal F\subseteq\CPTP(H\to K)$, and $0<\eps<1$, fix
a reference system $R\simeq H$. Write
$\rho_{\mathcal T}^\psi=(\id_R\otimes\mathcal T)(\ket\psi\bra\psi)$
and $\sigma_\M^\psi=(\id_R\otimes\M)(\ket\psi\bra\psi)$.
Then
\begin{equation}
 \begin{aligned}
 &\inf_{\substack{\ket\psi,Q:\ \langle\psi|\psi\rangle=1,\ 0\le Q\le I\\
       \Tr Q\rho_{\mathcal T}^\psi\ge1-\eps}}
       \sup_{\M\in\mathcal F}\Tr Q\sigma_\M^\psi\\
 &\hspace{3em}=\max_{\M\in\mathcal F}b_\eps(\mathcal T,\M).
 \end{aligned}
 \label{eq:testing-minimax}
\end{equation}
The infimum ranges over input states and acceptance effects satisfying
the target acceptance constraint. The same pair $(\ket\psi,Q)$ is
used for every $\M\in\mathcal F$, and an optimal pair exists.
\end{lemma}
\begin{proof}
Write the input state as
$\ket\psi=(I_H\otimes\sqrt\rho)\ket{\Gamma_H}$ for some
$\rho\in\Dens(H)$, with the reference isometry absorbed into $Q$.
Its physical input marginal is $\rho$. Using
$(I_H\otimes\sqrt\rho)\ket{\Gamma_H}
=({\sqrt{\rho}}^T\otimes I_H)\ket{\Gamma_H}$,
we can move the input operator to the reference factor.
For any channel $\Phi$, the output is
$({\sqrt{\rho}}^T\otimes I_K)C_\Phi({\sqrt{\rho}}^T\otimes I_K)$.
Define
\[
 T=({\sqrt{\rho}}^T\otimes I_K)Q({\sqrt{\rho}}^T\otimes I_K).
\]
The bounds $0\le Q\le I$ imply $0\le T\le\rho^T\otimes I_K$,
and cyclicity of the trace gives the acceptance probability
$\Tr TC_\Phi$.

Conversely, given $\rho\in\Dens(H)$ and $0\le T\le\rho^T\otimes I_K$,
use the input above and set
$Q=((\rho^T)^{-1/2}\otimes I_K)T((\rho^T)^{-1/2}\otimes I_K)$,
where the inverse is taken on the support and $Q$ is extended by zero.
The bound on $T$ ensures that it is supported on
$\supp(\rho^T)\otimes K$ and that $0\le Q\le I$, so this choice
recovers $T$ and realizes the same acceptance probabilities.
Thus testers are equivalently represented by
\begin{equation}
 \begin{gathered}
 \rho\in\Dens(H),\qquad 0\le T\le\rho^T\otimes I_K,\\
 \Pr(\text{accept}\mid\Phi)=\Tr TC_\Phi.
 \end{gathered}
 \label{eq:tester}
\end{equation}

The feasible set $\mathcal X_\eps$ of pairs $(\rho,T)$ satisfying
\eqref{eq:tester} and $\Tr TC_\mathcal T\ge1-\eps$ is nonempty,
compact, and convex. Nonemptiness follows by taking
$T=(1-\eps)\rho^T\otimes I_K$ for any $\rho\in\Dens(H)$.
The constraints are closed and convex, and $0\le T\le I_{H\otimes K}$
ensures boundedness. Since $\Tr TC_\M$ is continuous and bilinear,
Sion's minimax theorem \cite{Sion} gives
\[
 \min_{(\rho,T)\in\mathcal X_\eps}\max_{\M\in\mathcal F}\Tr TC_\M
 =\max_{\M\in\mathcal F}\min_{(\rho,T)\in\mathcal X_\eps}\Tr TC_\M.
\]
The inner minimum on the right is $b_\eps(\mathcal T,\M)$,
proving \eqref{eq:testing-minimax}. Compactness and continuity
ensure attainment of the tester minimum; the maximum over
$\mathcal F$ is attained because the inner minimum is upper
semicontinuous in $\M$. The realization above gives an optimal
input and effect with $R\simeq H$.
\end{proof}

\subsection{Uniform auxiliary map approximation}

For channels $\mathcal T,\mathcal M:H\to K$ and $t\ge0$, put
\[
 h_t(\mathcal T\Vert\mathcal M)
 =\sup_{\ket\psi,\,0\le Q\le I}(p-tq),
\]
where $p,q$ are the two acceptance probabilities for a common
reference-assisted input and effect. The following one-shot statement
keeps both auxiliary spaces fixed. This is essential for subtracting
the auxiliary maps and representing the residual.
The approximation bound was established in
\cite[Prop.~4.2]{GaoJiLiu}; we include its proof and the domination
consequences used below for completeness.

\begin{lemma}[Testing and auxiliary map approximation]
\label{lem:auxiliary-map-approximation}
Fix Stinespring isometries $V:H\to K\otimes E$ and
$W:H\to K\otimes F$ of channels $\mathcal T$ and $\mathcal M$.
For every $t\ge0$ there is an auxiliary linear map $A:F\to E$ with
\begin{equation}
 \opnorm A\le\sqrt t,\qquad
 \opnorm{V-(I_K\otimes A)W}^2
       \le h_t(\mathcal T\Vert\mathcal M).
 \label{eq:testing-auxiliary-map}
\end{equation}
In particular, if a channel $\mathcal L\cple t\mathcal M$ satisfies
$\dnorm{\mathcal L-\mathcal T}\le\delta$, the squared error in
\eqref{eq:testing-auxiliary-map} is at most $\delta/2$.
If $\mathcal T\cple t\mathcal M$, one can take zero error.
Conversely, the CP map obtained from $(I_K\otimes A)W$ by tracing
out the auxiliary system $E$ is dominated by $\opnorm A^2\mathcal M$.
\end{lemma}

\begin{proof}
The tester representation \eqref{eq:tester} and semidefinite duality
give, with $\Delta=C_{\mathcal T}-tC_{\mathcal M}$,
\begin{equation*}
 \begin{aligned}
 h_t(\mathcal T\Vert\mathcal M)
 &=\max_{\rho\in\Dens(H),\,0\le T\le\rho^T\otimes I_K}
             \Tr\Delta T\\
 &=\min_{Y\ge0,\,Y\ge\Delta}\opnorm{\Tr_KY}.
 \end{aligned}
\end{equation*}
This is an instance of the semidefinite approach to completely bounded
norms \cite{WatrousSDP}.
For completeness, a positive multiplier $Y$ for the tester upper
bound and a scalar $\mu$ for $\Tr\rho=1$ give the dual constraints
$Y\ge\Delta$ and $(\Tr_KY)^T\le\mu I$.
Strict feasibility follows by choosing a faithful $\rho$, a sufficiently
small positive scalar tester, and, on the dual side, sufficiently large
scalar $Y,\mu$. A dual sublevel set is compact because
$\Tr Y\le(\dim H)\opnorm{\Tr_KY}$, so a minimizing $Y$ exists.

Let $F_V,F_W$ be the matrices whose columns are the vectorized Kraus
operators of the prescribed dilations. Thus
$F_VF_V^\dagger=C_{\mathcal T}$ and
$F_WF_W^\dagger=C_{\mathcal M}$. Set
$G=tC_{\mathcal M}+Y$ and let $G^+$ denote its Moore--Penrose inverse.
Define
\[
 F_0=tC_{\mathcal M}G^+F_V,\qquad F_1=YG^+F_V.
\]
Since $C_{\mathcal T}\le G$, every column of $F_V$ lies in $\supp G$.
Thus $F_0+F_1=GG^+F_V=F_V$.
Put $Z_0=tC_{\mathcal M}$ and $Z_1=Y$. For $j=0,1$,
\[
 F_jF_j^\dagger
 =Z_jG^+C_{\mathcal T}G^+Z_j
 \le Z_jG^+Z_j\le Z_j.
\]
The first inequality uses $C_{\mathcal T}\le G$; the second follows
from $0\le Z_j\le G$. To justify the latter even for singular $G$,
the operator $(G^+)^{1/2}Z_j^{1/2}$ is a contraction, because its
product with its adjoint is at most the projection onto $\supp G$.
Hence $Z_j^{1/2}G^+Z_j^{1/2}\le I$, and congruence by
$Z_j^{1/2}$ gives $Z_jG^+Z_j\le Z_j$.
In particular,
$F_0F_0^\dagger\le tF_WF_W^\dagger$ and $F_1F_1^\dagger\le Y$.

The first bound implies that the columns of $F_0$ lie in the range
of $F_W$. Taking $D=F_W^+F_0$ therefore gives $F_0=F_WD$, and
\[
 DD^\dagger
 \le tF_W^+F_WF_W^\dagger(F_W^+)^\dagger
 =tP_{(\ker F_W)^\perp}\le tI_F.
\]
Here $P_{(\ker F_W)^\perp}$ is the orthogonal projection onto the
indicated subspace of $F$.
With the column convention above, $A=D^T:F\to E$ is the
auxiliary map and $\opnorm A=\opnorm D\le\sqrt t$.
Let $R=V-(I_K\otimes A)W$. Its Kraus-column matrix is $F_1$, so
$\Tr_KF_1F_1^\dagger=(R^\dagger R)^T$ in our Choi convention.
Consequently,
\[
 \opnorm R^2=\opnorm{\Tr_KF_1F_1^\dagger}
 \le\opnorm{\Tr_KY}
 =h_t(\mathcal T\Vert\mathcal M).
\]
This proves \eqref{eq:testing-auxiliary-map} in the prescribed
auxiliary spaces, with one auxiliary map $A$ for all inputs.

If $\mathcal L\cple t\mathcal M$, every tester satisfies
$p-tq\le p-p_{\mathcal L}\le\delta/2$, since the two target
outputs have equal trace. For exact domination take $\mathcal L=\mathcal T$.
Finally, the Kraus-column matrix of $(I_K\otimes A)W$ is $F_WA^T$;
its product with its adjoint is bounded by
$\opnorm A^2 C_{\mathcal M}$, proving the converse assertion.
\end{proof}

\subsection{Completion of a trace defect}
We will also use a completion estimate in which the free mixture is
weighted by the actual trace defect.

\begin{lemma}[Completion at the cost of the trace defect]
\label{lem:defect-completion}
Assume only that $\F_n$ is convex and contains
$\R_\omega^{\otimes n}$.
Let $\mathcal Q$ be CP and trace nonincreasing, and suppose
$\mathcal Q\cple\lambda\mathcal M$ for $\mathcal M\in\F_n$
and $\lambda\ge1$. Put $D=I-\mathcal Q^\dagger(I)$ and
$d=\opnorm D$. Then there are a channel $\mathcal L$ and
$\mathcal S\in\F_n$ with
\begin{equation*}
 \mathcal L\cple(\lambda+d)\mathcal S,\qquad
 \dnorm{\mathcal L-\N_n}
 \le2\dnorm{\mathcal Q-\N_n}.
\end{equation*}
Here $0\le d\le\min\{1,\dnorm{\mathcal Q-\N_n}\}$.
\end{lemma}

\begin{proof}
Set $\mathcal C_D(X)=\Tr(DX)\omega^{\otimes n}$ and
$\mathcal L=\mathcal Q+\mathcal C_D$. Then $\mathcal L$ is CPTP,
$\mathcal C_D\cple d\R_\omega^{\otimes n}$, and
$\dnorm{\mathcal C_D}=d$.
The free channel
\[
 \mathcal S=
 \frac{\lambda\mathcal M+d\R_\omega^{\otimes n}}{\lambda+d}
\]
gives the claimed domination by convexity. For every input state $\rho$,
$\Tr D\rho=\Tr(\N_n-\mathcal Q)(\rho)$ is at most
$\dnorm{\N_n-\mathcal Q}$. Maximizing over $\rho$ bounds $d$;
the triangle inequality proves the error estimate.
\end{proof}

\section{Exponential approximation and the Stein identity}
\label{sec:exponential}
We prove \cref{thm:main,thm:exponential-completion} under the three
conditions in \cref{def:axioms}, using the testing and approximation
tools of \cref{sec:one-shot}.
For a fixed testing tolerance, we start from the lower limit
$b_\eps=\liminf_n a_{\eps,n}/n$, without assuming convergence of either
the testing exponent or $E_n/n$.

The proof has two amplification steps. First,
\cref{prop:free-amplification} turns any unbounded sequence of
approximations with vanishing error into exponentially accurate
approximations at every sufficiently large blocklength, with an
arbitrarily small increase in rate. Next,
\cref{prop:iterative-amplification} uses testing to construct such a
sequence at a rate below a previously available approximation rate.
Starting from exact replacer domination and iterating this improvement
gives exponential approximation at every rate above $b_\eps$.
These approximations force the testing upper limits at all tolerances
to be at most $b_\eps$. The entropy comparison in
\cref{lem:entropy-comparison} then identifies the common ordinary limit.

\subsection{Amplification and exact padding}

The following proposition starts from CPTP approximations to
$\N^{\otimes k}$ at an unbounded sequence of blocklengths $k$,
with vanishing diamond error and a bounded domination rate,
as specified in \eqref{eq:amplification-input}.
Under convexity, tensor closure, and the faithful-replacer
assumption, it produces exponentially accurate approximations at
every sufficiently large blocklength.
The proof of \cref{prop:iterative-amplification} in the next
subsection constructs the required sequence by combining
an auxiliary map approximation obtained from testing with an
approximation already available at a higher domination rate.
That construction uses only \cref{def:axioms}.
The tensor truncation below adapts the fixed-block amplification of
\cite[Lemma~4.5]{GaoJiLiu} to free dominators and exact trace preservation.

\begin{proposition}[Amplification from vanishing error]
\label{prop:free-amplification}
Suppose each $\F_n$ is convex and $\F$ satisfies
\ref{ax:tensor} and \ref{ax:replacer}. Fix a finite $r_0\ge0$.
Assume that, along an unbounded sequence of blocklengths $k$, there
are CPTP maps $\mathcal L_k$, free channels $\mathcal M_k\in\F_k$,
and numbers $\lambda_k\ge1$ such that
\begin{equation}
 \begin{gathered}
 \mathcal L_k\cple\lambda_k\mathcal M_k,\qquad
 \delta_k:=\dnorm{\mathcal L_k-\N_k}\longrightarrow0,\\
 \limsup_{k\to\infty}\frac{\log_2\lambda_k}{k}\le r_0.
 \end{gathered}
 \label{eq:amplification-input}
\end{equation}
The limits in \eqref{eq:amplification-input} are taken along that sequence.
Then for every $S>r_0$ there are free dominators and CPTP
approximations satisfying \eqref{eq:exponential-completion}.
\end{proposition}

\begin{proof}
Write $C=\log_2 b-\log_2 w$. Equation~\eqref{eq:replacer-domination},
applied to a maximally entangled input, gives
$\N_k\cple2^{kC}\R_\omega^{\otimes k}$.
If $S\ge C$, the exact target and the free replacer prove the claim.
Otherwise choose $r_0<r<S<C$ and set $L=C+1$.

\emph{Step 1: Two auxiliary maps between the same spaces.}
Fix a Stinespring isometry $V:A\to B\otimes E$ of $\N$.
For each sufficiently large $k$ in the given sequence, form the free channel
\[
 \widehat{\mathcal M}_k
 =\tfrac12\mathcal M_k+\tfrac12\R_\omega^{\otimes k}
\]
and fix a dilation $W_k:A^k\to B^k\otimes F_k$ of it.
By \cref{lem:auxiliary-map-approximation}, there are auxiliary maps
$A_k,B_k:F_k\to E^{\otimes k}$ such that
\begin{equation*}
 \begin{gathered}
 X_k=(I\otimes A_k)W_k,\qquad
 V^{\otimes k}=(I\otimes B_k)W_k,\\
 \opnorm{V^{\otimes k}-X_k}\le e_k:=\sqrt{\delta_k/2},\\
 \opnorm{A_k}\le2^{kr/2},\qquad
 \opnorm{B_k}\le2^{kL/2}.
 \end{gathered}
\end{equation*}
For $A_k$ use $\mathcal L_k\cple2\lambda_k\widehat{\mathcal M}_k$
and $\log_2(2\lambda_k)/k<r$ eventually.
For $B_k$ use
$\N_k\cple2^{kC+1}\widehat{\mathcal M}_k$ and $kC+1\le kL$.
Thus $B_k-A_k$ represents the exact residual in the same codomain,
even though the dominator may vary with $k$.

\emph{Step 2: Fix the block and truncate its exact expansion.}
Choose $0<\eta<1$ such that $\bar S=r+\eta(L-r)<S$,
and put $z=4^{1/\eta}$. Choose and then fix $k$ in that sequence large enough that
\begin{equation*}
 e_k\le\frac1{1+z},\qquad
 v:=\bar S+\frac{2\log_2 3}{k}<S.
\end{equation*}
Keep this $k$ fixed for the remainder of the proof; the number $m$
of repeated blocks will tend to infinity.
Abbreviate $X=X_k$ and $H=V^{\otimes k}-X$.
The identity $(X+H)^{\otimes m}=V^{\otimes km}$ is exact.
Let $Z_m$ be its subset expansion restricted to terms with at most
$\lfloor\eta m\rfloor$ factors $H$. Since
$\opnorm X\le1+e_k$ and $\opnorm H\le e_k$, the omitted tail obeys
\begin{align*}
 \opnorm{Z_m-V^{\otimes km}}
 &\le\sum_{j>\eta m}\binom mj(1+e_k)^{m-j}e_k^j\\
 &\le z^{-\eta m}(1+e_k+ze_k)^m\le2^{-m}.
\end{align*}
The second inequality follows by multiplying each omitted term by
$z^{j-\eta m}\ge1$ and including all terms in the binomial sum.

Replace $X$ and $H$ in the same subset sum by $A_k$ and $B_k-A_k$.
This defines $D_m:F_k^{\otimes m}\to E^{\otimes km}$ with
$Z_m=(I\otimes D_m)W_k^{\otimes m}$ and
\begin{align*}
 \opnorm{D_m}
 &\le\sum_{j\le\eta m}\binom mj2^j
                 2^{k[(m-j)r+jL]/2}\\
 &\le3^m2^{km\bar S/2}=2^{kmv/2}.
\end{align*}
We used $\opnorm{B_k-A_k}\le2\,2^{kL/2}$ and $r<L$.
The tensor-power dominator $\widehat{\mathcal M}_k^{\otimes m}$
belongs to $\F_{km}$, although it may contain arbitrary correlations
inside each block.

\emph{Step 3: Pad exactly and restore trace preservation.}
For $n=km+\ell$, $0\le\ell<k$, use
\[
 \mathcal T_n=\widehat{\mathcal M}_k^{\otimes m}
                   \otimes\R_\omega^{\otimes\ell}\in\F_n.
\]
For $\ell=0$ the last factor is omitted. The exact case of
\cref{lem:auxiliary-map-approximation} applied to the replacer gives
a representation of $V^{\otimes\ell}$ with auxiliary map norm
at most $2^{\ell C/2}$.
Tensoring $D_m$ with this auxiliary map yields a representation of
$\widehat V_n$ satisfying
\[
 \opnorm{\widehat V_n-V^{\otimes n}}\le e_n:=2^{-m}.
\]
The corresponding auxiliary map relative to $\mathcal T_n$ has squared
operator norm at most $\lambda_n:=2^{kmv+\ell C}$.
Exact padding does not increase the error.
The contraction $T_n=\widehat V_n/(1+e_n)$ defines the subchannel
$\mathcal Q_n(X)=\Tr_{E^n}T_nXT_n^\dagger$.
It satisfies $\mathcal Q_n\cple\lambda_n\mathcal T_n$ and
\[
 \opnorm{T_n-V^{\otimes n}}\le2e_n,\qquad
 \dnorm{\mathcal Q_n-\N_n}\le4e_n,
\]
by \eqref{eq:single-kraus-continuity} and contraction of the partial trace.
\Cref{lem:defect-completion} gives CPTP $\mathcal L_n$ and free
$\mathcal S_n$ with
\[
 \mathcal L_n\cple(\lambda_n+1)\mathcal S_n,\qquad
 \dnorm{\mathcal L_n-\N_n}\le8e_n.
\]
Since $k$ is fixed and $v<S$,
$\log_2(\lambda_n+1)\le nv+\ell(C-v)+1\le nS$ eventually.
Finally $8e_n=8\,2^{-\lfloor n/k\rfloor}
\le16e^{-(\ln2)n/k}$. This proves the claim for every sufficiently
large integer $n$, with $K_S=16$ and $\gamma_S=(\ln2)/k$.
\end{proof}

\subsection{Iterative reduction of the approximation rate}

Testing gives a uniform auxiliary map approximation with error
strictly below one, which need not tend to zero. The next argument
combines it with an exponentially accurate approximation at a higher
rate. The result improves that higher rate by a fixed fraction.
Iteration reaches every rate above the lower limiting testing exponent.
This adapts the two-rate tensor amplification of
\cite[Lemma~4.4]{GaoJiLiu} to the present composite setting.

\begin{proposition}[From testing to exponential approximation]
\label{prop:iterative-amplification}
Assume \ref{ax:convex}, \ref{ax:tensor}, and \ref{ax:replacer}.
For any fixed $0<\eps<1$, put
$b_\eps=\liminf_n a_{\eps,n}/n$. Then for every $S>b_\eps$
there are free dominators and CPTP approximations satisfying
\eqref{eq:exponential-completion}.
\end{proposition}
\begin{proof}
Let $P(s)$ denote the existence of the approximations in
\eqref{eq:exponential-completion} at rate $s$, for all sufficiently
large blocklengths. Exact replacer domination proves $P(C)$.
If $S\ge C$ the conclusion is immediate. Otherwise fix
$b_\eps<r<S<C$ and an increasing sequence of integers $k$ along
which $a_{\eps,k}/k\to b_\eps$.

\emph{Step 1: A rough approximation from testing.}
Choose a hardest $\mathcal M_k\in\F_k$ from
\cref{lem:testing-minimax}, and set $t_k=\eps/\beta_{\eps,k}$.
Every tester with target acceptance $p\ge1-\eps$ has alternative
acceptance $q\ge\beta_{\eps,k}$, hence $p-t_kq\le1-\eps$.
For $p<1-\eps$ the same upper bound holds because $q\ge0$.
Thus
\begin{equation}
 h_{t_k}(\N_k\Vert\mathcal M_k)\le1-\eps.
 \label{eq:rough-auxiliary-map-testing}
\end{equation}
Write $d=\sqrt{1-\eps}<1$ and fix a dilation $V$ of $\N$.

\emph{Step 2: One improvement of a higher rate.}
Suppose $r<s\le C$ and $P(s)$ holds. Let
$\mathcal L_k\cple2^{ks}\mathcal T_k$ be its CPTP
approximations, with $\mathcal T_k\in\F_k$ and diamond error
at most $Ke^{-\gamma k}$.
For each sufficiently large $k$ in the selected sequence, fix a
dilation $W_k$ of the common free dominator
\[
 \mathcal U_k=\tfrac12\mathcal M_k+\tfrac12\mathcal T_k.
\]
Since $\mathcal U_k\cpge\mathcal M_k/2$, the testing tail at
coefficient $2t_k$ is at most the tail in
\eqref{eq:rough-auxiliary-map-testing}.
Apply \cref{lem:auxiliary-map-approximation} twice, with the
same prescribed $V^{\otimes k}$ and $W_k$, to obtain
\begin{equation*}
 \begin{gathered}
 X_k=(I\otimes A_k)W_k,\quad Y_k=(I\otimes B_k)W_k,\\
 \opnorm{V^{\otimes k}-X_k}\le d,\qquad
 \opnorm{V^{\otimes k}-Y_k}\le e_k,\\
 e_k:=\sqrt{K/2}e^{-\gamma k/2},\\
 \opnorm{A_k}\le\sqrt{2t_k}\le2^{kr/2},\quad
 \opnorm{B_k}\le\sqrt2\,2^{ks/2}.
 \end{gathered}
\end{equation*}
The bound on $A_k$ holds eventually by the choice of $r$.
Put $H_k=Y_k-X_k$, $A_0=1+d$, and $B_0=(1+d)/2<1$.
For sufficiently large $k$,
\[
 \opnorm{X_k}\le A_0,\qquad \opnorm{H_k}\le d+e_k\le B_0.
\]
Choose $0<\eta<1$ and $\kappa>0$, depending only on $d$, so that
\begin{equation}
 h_2(x)\ln2+(1-x)\ln A_0+x\ln B_0\le-\kappa
 \quad(\eta\le x\le1).
 \label{eq:iterative-tail-choice}
\end{equation}
Such a choice exists by continuity, since the expression equals
$\ln B_0<0$ at $x=1$.

Take $m=k$ and truncate $(X_k+H_k)^{\otimes m}=Y_k^{\otimes m}$
to terms with at most $\lfloor\eta m\rfloor$ residual factors.
Denote the resulting operator by $Z_k$. The bound
$\binom mj\le2^{mh_2(j/m)}$ gives
\begin{align*}
 \opnorm{Z_k-V^{\otimes km}}
 &\le(m+1)e^{-\kappa m}
       +me_k(1+e_k)^{m-1}\\
 &=:u_k\longrightarrow0.
\end{align*}
The second term is the telescoping bound for
$Y_k^{\otimes m}-V^{\otimes km}$; it vanishes because $m=k$
and $e_k$ decays exponentially.

Replacing $X_k,H_k$ by $A_k,B_k-A_k$ in this same subset sum
gives an auxiliary map $D_k$ relative to $W_k^{\otimes m}$.
Set $\bar s=(1-\eta)r+\eta s$. Since $r<s$ and
$\opnorm{B_k-A_k}\le3\,2^{ks/2}$,
\begin{align*}
 \opnorm{D_k}
 &\le\sum_{j\le\eta m}\binom mj3^j
             2^{k[(m-j)r+js]/2}\\
 &\le4^m2^{km\bar s/2}.
\end{align*}
Tensor closure gives $\mathcal U_k^{\otimes m}\in\F_{km}$.
Scale $Z_k$ by $1+u_k$ to get a contraction, then apply
\cref{lem:defect-completion}, exactly as in the proof of
\cref{prop:free-amplification}. This gives CPTP approximations
at lengths $N=k^2$, with error at most $8u_k\to0$ and domination
cost at most
\[
 \lambda_N=2^{N\bar s+4k}+1.
\]
In particular $\limsup (\log_2\lambda_N)/N\le\bar s$ along
this unbounded sequence. By \cref{prop:free-amplification},
$P(s')$ holds for every $s'>\bar s$.
Here both $k$ and $m=k$ grow, and the truncation alone supplies
vanishing error only along the lengths $N=k^2$.
The application of \cref{prop:free-amplification} then fixes one
accurate block and supplies exponential accuracy at all sufficiently
large lengths, including those outside this sequence.

\emph{Step 3: Iterate a fixed contraction of the rate gap.}
The choice of $\eta$ in \eqref{eq:iterative-tail-choice} depends only
on $d=\sqrt{1-\eps}$, so the same $\eta$ works at every iteration.
Set $\theta=(1+\eta)/2\in(\eta,1)$ and
$s_j=r+\theta^j(C-r)$. Starting from $P(s_0)=P(C)$,
Step 2 proves $P(s_{j+1})$ because
\[
 s_{j+1}>(1-\eta)r+\eta s_j.
\]
After finitely many steps, $s_j<S$, proving $P(S)$.
All constants may depend on the chosen rate margin. No quantitative
completion estimate, permutation average, or marginal reduction
is used in this argument.
\end{proof}

\subsection{Entropy from an accurate free-dominated channel}

The following comparison identifies the entropy rate without symmetry
or an entropy minimax argument. Mixing the dominator with the replacer
gives a bound on $\log_2\sigma-\log_2(\mu\otimes I)$, where $\sigma$
is the resulting reference--output state and $\mu$ its reference marginal.
On $\supp\mu\otimes B^n$, this bound is uniform over input states.

\begin{lemma}[Entropy comparison]\label{lem:entropy-comparison}
Suppose $\mathcal L\cple2^z\mathcal S$ for a CPTP map
$\mathcal L$ and $\mathcal S\in\F_n$, and
$\dnorm{\mathcal L-\N_n}\le\delta\le1$. Then
\begin{equation*}
 E_n\le z+1+\frac\delta2
       \bigl[n(C+\log_2(ab))+1\bigr]+h_2(\delta/2).
\end{equation*}
Only convexity of $\F_n$ and membership of
$\R_\omega^{\otimes n}$ in $\F_n$ are needed.
\end{lemma}
\begin{proof}
Set $\mathcal M=(\mathcal S+\R_\omega^{\otimes n})/2\in\F_n$.
On a common pure input with $R\simeq A^n$, let $\rho,\nu,\sigma$
be the outputs of $\N_n,\mathcal L,\mathcal M$, respectively,
and let $\mu$ be their common reference marginal.
The domination $\nu\le2^{z+1}\sigma$ gives
$D(\nu\Vert\sigma)\le z+1$.
On $\supp\mu\otimes B^n$, both logarithms in
$H=\log_2\sigma-\log_2(\mu\otimes I)$ are well-defined.
The replacer component and \eqref{eq:dimension-order} imply
\[
 \tfrac12w^n\mu\otimes I\le\sigma\le b^n\mu\otimes I.
\]
Operator monotonicity of the logarithm therefore gives
\[
 (n\log_2 w-1)I\le H\le n\log_2 b\,I.
\]
Put $t=\trnorm{\rho-\nu}/2\le\delta/2$.
Since $\rho-\nu$ has zero reference marginal and zero trace,
\[
 |\Tr(\rho-\nu)\log_2\sigma|
 =|\Tr(\rho-\nu)H|\le t(nC+1).
\]
Here a trace-zero Hermitian matrix of trace norm $2t$ pairs with
an operator of spectral width $nC+1$ by at most $t(nC+1)$.
Entropy continuity in dimension at most $(ab)^n$ gives
$|S(\rho)-S(\nu)|\le tn\log_2(ab)+h_2(t)$
\cite[Thm.~5.26]{Watrous}. This also holds trivially when $ab=1$.
Subtracting the two relative entropies and using $t\le\delta/2\le1/2$
proves the displayed bound uniformly over the input. Take the
supremum over inputs and then use $E_n\le D_\ch(\N_n\Vert\mathcal M)$.
\end{proof}

\subsection{Identification of the asymptotic rate}\label{sec:limits}

\begin{proof}[Proof of \cref{thm:main,thm:exponential-completion}]
\emph{Step 1: From a lower limit to a common ordinary limit.}
Fix $0<\eps<1$ and put $b_\eps=\liminf_n a_{\eps,n}/n\in[0,C]$.
For every $S>b_\eps$, \cref{prop:iterative-amplification} gives
free $\mathcal S_n$ and CPTP $\mathcal L_n$, for all sufficiently
large $n$, with domination $2^{nS}$ and error at most
$d_n=K_Se^{-\gamma_S n}$.
Fix $0<\eps'<1$ and a tester $(R,\ket\psi,Q)$ with target
acceptance at least $1-\eps'$. Let $\rho=\rho_n^\psi$ be its
target output and let
$\widetilde\rho=(\id_R\otimes\mathcal L_n)(\ket\psi\bra\psi)$.
The diamond bound gives $\trnorm{\widetilde\rho-\rho}\le d_n$.
Both outputs have trace one, so their difference is traceless.
For an acceptance effect $0\le Q\le I$, this gives the sharper
probability bound
\[
 \Tr Q\widetilde\rho
 \ge\Tr Q\rho-\tfrac12\trnorm{\widetilde\rho-\rho}
 \ge1-\eps'-d_n/2.
\]
Moreover, $\mathcal L_n\cple2^{nS}\mathcal S_n$ implies
$\widetilde\rho\le2^{nS}\sigma_{\mathcal S_n}^\psi$ for this
reference-assisted input. Since $\mathcal S_n\in\F_n$, the
tester's largest acceptance probability over free alternatives obeys
\[
 \begin{aligned}
 q_n&:=\sup_{\M\in\F_n}\Tr Q\sigma_\M^\psi\\
    &\ge\Tr Q\sigma_{\mathcal S_n}^\psi
     \ge2^{-nS}\Tr Q\widetilde\rho
     \ge2^{-nS}(1-\eps'-d_n/2).
 \end{aligned}
\]
Minimizing this uniform bound over feasible testers gives
\[
 \beta_{\eps',n}\ge2^{-nS}(1-\eps'-d_n/2).
\]
For fixed $S$ and $0<\eps'<1$, we have $d_n\to0$, so
$1-\eps'-d_n/2\to1-\eps'>0$.
The factor is therefore positive for all sufficiently large $n$,
and its logarithm converges to the finite constant
$\log_2(1-\eps')$. Dividing this bounded logarithm by $n$ gives
$n^{-1}\log_2(1-\eps'-d_n/2)\to0$.
To apply this observation, take negative logarithms of the testing
bound and use $a_{\eps',n}=-\log_2\beta_{\eps',n}$:
\[
 \frac{a_{\eps',n}}n
 \le S-\frac1n\log_2(1-\eps'-d_n/2).
\]
The second term on the right tends to zero, so
$\limsup_n a_{\eps',n}/n\le S$.
Since this holds for every $S>b_\eps$, letting $S\downarrow b_\eps$
yields
\begin{equation*}
 \limsup_n\frac{a_{\eps',n}}n
 \le\liminf_n\frac{a_{\eps,n}}n
 \qquad(0<\eps,\eps'<1).
\end{equation*}
Taking $\eps'=\eps$ proves ordinary convergence at each tolerance;
interchanging the two tolerances shows that all limits coincide. Hence
\begin{equation*}
 \frac{a_{\eps,n}}n\longrightarrow R
 \qquad(0<\eps<1).
\end{equation*}
The same construction already gives exponential free approximation
at each $S>R$. Neither this convergence argument nor the approximation
construction assumes additivity of channel relative entropy or
near subadditivity of the testing quantity.

\emph{Step 2: Identification of the entropy limit.}
The weak converse and \cref{lem:entropy-comparison} identify the common limit.
The weak converse \eqref{eq:weak-converse} gives
$\liminf_n E_n/n\ge(1-\eps)R$ for every fixed $\eps$.
Letting $\eps\downarrow0$ after this limit yields
$\liminf_n E_n/n\ge R$.
Conversely, for each $S>R$ the approximations from Step 1 have
$z=nS$ and exponentially vanishing $\delta=d_n$ in
\cref{lem:entropy-comparison}. Consequently
$\limsup_n E_n/n\le S$. Let $S\downarrow R$.
Thus the ordinary entropy limit exists and equals $R$.
The bound $0\le R\le C$ follows from \cref{lem:bounds}.

\Cref{thm:exponential-completion} follows from the approximation
already obtained in Step 1 and the identification in Step 2.
\end{proof}

\section{Asymptotic consequences}\label{sec:consequences}
The Stein identity and exponential approximation give three
consequences under the basic assumptions in \cref{def:axioms}:
asymptotic equipartition for trace-preserving smoothing, an
exponential strong converse, and stability under vanishing
diamond-norm perturbations.

\subsection{CPTP-smoothed asymptotic equipartition}

Classically, the smallest constant $\lambda$ such that
$P(x)\le\lambda Q(x)$ for every $x$ is the largest likelihood ratio.
Its logarithm is max-relative entropy. For quantum states, positive
operator order replaces pointwise order:
\[
 D_{\max}(\rho\Vert\sigma)
       =\log_2\inf\{\lambda\ge1:\rho\le\lambda\sigma\}.
\]
The factor $2^{D_{\max}(\rho\Vert\sigma)}$ bounds every measurement acceptance ratio
\cite{Datta}. For channels, complete positivity imposes the same
bound simultaneously on every reference-assisted input. Define
the optimized max-relative entropy of a channel by
\begin{equation}
 D_{\max}(\mathcal L\Vert\F_n)
 =\inf_{\mathcal S\in\F_n}\log_2
       \inf\{\lambda\ge1:\mathcal L\cple\lambda\mathcal S\}.
 \label{eq:resource-max}
\end{equation}
The restriction $\lambda\ge1$ follows already by taking traces
on a normalized input. If $\mathcal L\cple\lambda\mathcal S$
and $\lambda>1$, the map
$\mathcal R=(\lambda\mathcal S-\mathcal L)/(\lambda-1)$ is a
channel, so
\[
 \mathcal S=\lambda^{-1}\mathcal L+(1-\lambda^{-1})\mathcal R.
\]
Thus $2^{-D_{\max}}$ describes the largest possible channel weight
inside a free dominator; the residual channel need not be free.
For $\delta\ge0$ define its CPTP-smoothed value
\begin{equation}
 Z_n^\delta=
 \inf_{\substack{\mathcal L\in\CPTP(A^n\to B^n)\\
                  \dnorm{\mathcal L-\N_n}\le\delta}}
             D_{\max}(\mathcal L\Vert\F_n).
 \label{eq:smoothed-resource-max}
\end{equation}
Smoothing means allowing approximation in the specified metric,
not averaging channels or smoothing each output state separately.
The same map $\mathcal L$ must approximate $\N_n$ on every
reference-assisted input. All feasible approximations here are
trace preserving.
The quantities in \eqref{eq:resource-max} and
\eqref{eq:smoothed-resource-max} are nonnegative and finite:
the Choi operator of every channel is positive with trace $a^n$,
so $C_{\mathcal L}\le a^nI$, whereas
$C_{\R_\omega^{\otimes n}}\ge w^n I$.
Thus $\mathcal L\cple(a/w)^n\R_\omega^{\otimes n}$.
These infima are attained. Indeed, the domination constraint is
closed in Choi operators, and the displayed uniform bound allows
$\lambda$ to be restricted to the compact interval $[1,(a/w)^n]$.
The free-channel set and the diamond ball of CPTP approximations
are compact as well.

The terminology asymptotic equipartition originates in classical
typicality: for independent samples from $P$, the quantity
$-n^{-1}\log_2P(X_1,\ldots,X_n)$ converges in probability to
the Shannon entropy. In relative-entropy formulations, smoothing
removes the effect of exceptional likelihood ratios, and the
normalized smoothed maximum converges to an entropy rate. The
next theorem establishes this form of equipartition for the
channel resource quantity in \eqref{eq:smoothed-resource-max}.

\begin{theorem}[CPTP-smoothed resource AEP]\label{thm:resource-AEP}
Under the assumptions of \cref{thm:main}, if
$0<\delta_n\le\bar\delta<2$ and
$\log_2(1/\delta_n)=o(n)$, then
\begin{equation}
 \lim_{n\to\infty}\frac{Z_n^{\delta_n}}n=R_\F(\N).
 \label{eq:resource-AEP}
\end{equation}
This includes all fixed radii in $(0,2)$, inverse-polynomial errors,
and $\delta_n=2^{-n^\theta}$ for every $0<\theta<1$.
\end{theorem}
\begin{proof}
For $0<\eps<1$ and $0\le\delta<2(1-\eps)$, suppose
$\mathcal L\cple2^z\mathcal S$ is feasible in
\eqref{eq:smoothed-resource-max}. Every test with target acceptance
at least $1-\eps$ accepts $\mathcal L$ with probability at least
$1-\eps-\delta/2$. The factor $1/2$ follows from exact trace
preservation and the use of the full diamond norm.
Its worst-case alternative acceptance is therefore at least
$2^{-z}(1-\eps-\delta/2)$. Minimizing over tests gives
$\beta_{\eps,n}\ge2^{-z}(1-\eps-\delta/2)$.
Taking logarithms and minimizing over the feasible approximations
gives the general testing lower bound
\begin{equation}
 a_{\eps,n}+\log_2(1-\eps-\delta/2)\le Z_n^\delta.
 \label{eq:AEP-lower}
\end{equation}

For the limit choose $\eps=(1-\bar\delta/2)/2>0$.
The lower logarithmic term is bounded uniformly in $n$, so
\cref{thm:main} gives $\liminf_n Z_n^{\delta_n}/n\ge R_\F(\N)$.
For any $S>R_\F(\N)$, \cref{thm:exponential-completion} gives
CPTP approximations at cost at most $nS$ and error at most
$K_Se^{-\gamma_S n}$. The condition on $\delta_n$ ensures that
$K_Se^{-\gamma_S n}\le\delta_n$ eventually. These approximations
are feasible, giving $\limsup_n Z_n^{\delta_n}/n\le S$.
Let $S\downarrow R_\F(\N)$.
\end{proof}

Three optimization problems should be distinguished. A prescribed
dominator $\M^{\otimes n}$ with trace-preserving smoothing, the same
dominator with trace-nonincreasing smoothing, and the free-family
optimization \eqref{eq:smoothed-resource-max} have different feasible
sets. Gour's counterexample \cite{Gour} concerns the first of these;
allowing a subchannel changes the normalization constraint. Here
we keep exact trace preservation and optimize the free dominator.
In the proof of this AEP, the optimization over the free dominator in
\eqref{eq:resource-max} permits tensor amplification and completion
of the trace defect.
The closure axioms construct a suitable dominator at each
blocklength, while the smoothing preserves the channel's exact
trace constraint. Equation~\eqref{eq:resource-AEP} identifies the
asymptotic resource cost of CP domination with the Stein rate.

\subsection{Operational threshold}
For an arbitrary $n$-use parallel tester, let $p_n$ denote its
acceptance probability on $\N_n$ and let $q_n$ denote its largest
acceptance probability over $\F_n$.
The pair $(p_n,q_n)$ describes its performance at any
type-I error.

\begin{theorem}[Operational threshold]\label{thm:threshold}
Under the assumptions of \cref{thm:main}, there exists a sequence of uniform testers with
\[
 p_n\longrightarrow1,\qquad
 -\frac1n\log_2q_n\longrightarrow R_\F(\N).
\]
Conversely, fix $r>R_\F(\N)$. There are $K_r,c_r>0$ such that
every parallel $n$-use tester satisfies
\begin{equation*}
 q_n\le2^{-nr}\quad\Longrightarrow\quad p_n\le K_r2^{-c_r n}.
\end{equation*}
The constants are independent of the tester and its reference.
In particular, exceeding the Stein rate forces the type-I error
to tend to one.
\end{theorem}
\begin{proof}
For each $j\ge2$, \cref{thm:main} supplies an $N_j$ such that
$|a_{1/j,n}/n-R_\F(\N)|\le1/j$ for all $n\ge N_j$.
Increase the $N_j$ so that they are strictly increasing and
$N_j\ge j$.
For $n\ge N_2$ use an optimal uniform tester at error
$1/j(n)$, where $j(n)=\max\{j:N_j\le n\}$.
The optimum exists by \cref{lem:testing-minimax};
$p_n\ge1-1/j(n)\to1$ and its exponent converges to $R_\F(\N)$.

For the converse, fix $R_\F(\N)<S<r$ and apply
\cref{thm:exponential-completion}. A diamond error at most
$d_n=K_Se^{-\gamma_S n}$ changes acceptance by at most $d_n/2$,
since both maps are trace preserving. CP domination therefore gives
\begin{equation*}
 p_n\le d_n/2+2^{nS}q_n .
\end{equation*}
If $q_n\le2^{-nr}$, this is at most
$(K_S/2+1)2^{-c_rn}$ for all sufficiently large $n$, with
$c_r=\min\{\gamma_S/\ln2,r-S\}>0$.
Increasing the prefactor covers the finitely many remaining $n$.
\end{proof}

\subsection{Stability under target perturbations}
The same exponent persists under asymptotically vanishing
perturbations of the target sequence.

\begin{corollary}[Diamond-norm stability]\label{cor:target-stability}
Assume the hypotheses of \cref{thm:main}.
Let $\mathcal T_n:\mathrm L(A^n)\to\mathrm L(B^n)$ be arbitrary CPTP maps with
$\dnorm{\mathcal T_n-\N_n}\to0$.
Let $\widehat\beta_{\eps,n}$ be defined as in \eqref{eq:beta},
with target $\mathcal T_n$ and the same alternatives $\F_n$.
For every fixed $0<\eps<1$,
\begin{equation*}
 \lim_n-\frac1n\log_2\widehat\beta_{\eps,n}=R_\F(\N).
\end{equation*}
The approximating targets $\mathcal T_n$ may have arbitrary
correlations across uses.
\end{corollary}
\begin{proof}
Choose any $u>0$ with $u<\min\{\eps,1-\eps\}$.
Eventually the diamond distance is at most $u$, so transfer of
the same input state and effect between the two targets gives
\[
 \beta_{\eps+u,n}\le
 \widehat\beta_{\eps,n}\le\beta_{\eps-u,n}.
\]
For the first inequality every test feasible for $\mathcal T_n$
is feasible for $\N_n$ at error $\eps+u$; the other direction
uses a test feasible for $\N_n$ at error $\eps-u$.
Both bounding exponents tend to $R_\F(\N)$ by \cref{thm:main}.
\end{proof}

\section{Quantitative finite-block completion}\label{sec:quantitative}

The asymptotic results give exponential approximation at any rate
strictly above $R_\F(\N)$. We now prove an explicit bound at a
specified blocklength and accuracy, with domination cost measured
against the finite-block testing quantity $a_{\eps,n}$.
This refinement uses permutation invariance and marginal closure.

\subsection{Additional assumptions}
Write $S_n$ for the permutation group of $[n]$. For $\pi\in S_n$,
let $U_\pi^H$ permute the factors of $H^n$ according to $\pi$.

\begin{definition}[Additional conditions for quantitative completion]
\label{def:quantitative-axioms}
In addition to \cref{def:axioms}, consider the following conditions:
\begin{enumerate}[label=\textup{(F\arabic*)},ref=F\arabic*,start=4]
\item\label{ax:permutation} For every $\pi\in S_n$,
$\Ad_{U_\pi^B}\circ\M\circ\Ad_{(U_\pi^A)^{\dagger}}\in\F_n$ whenever
$\M\in\F_n$.
\item\label{ax:marginal} There is a fixed positive-definite
$\tau\in\Dens(A)$ such that, for every $n\ge2$ and $\M\in\F_n$,
\begin{equation*}
 X\longmapsto\Tr_{B_n}\M(X\otimes\tau)
 \quad\hbox{belongs to }\F_{n-1}.
\end{equation*}
\end{enumerate}
Iteration of \ref{ax:marginal} gives
\begin{equation}
 \begin{gathered}
 X\longmapsto
 \Tr_{B_{m+1}\cdots B_n}\M(X\otimes\tau^{\otimes(n-m)})\in\F_m,\\
 1\le m<n.
 \end{gathered}
 \label{eq:iterated-marginal}
\end{equation}
Together with \ref{ax:permutation}, this gives the analogous statement
for every discarded subset, after relabeling the remaining sites in
their original order. The input state $\tau$ and output state
$\omega$ are fixed across blocklengths and may be different.
\end{definition}

\begin{remark}[The marginal assumption]\label{rem:marginal-scope}
Condition \ref{ax:marginal} permits reduction of a free block
channel by inserting $\tau$ and discarding its output.
We use it for weighted discarding and to replace selected inputs
and outputs in the free multiplier construction.
The asymptotic results in \cref{sec:exponential,sec:consequences}
hold under the three basic assumptions alone.
\end{remark}

\subsection{Free-channel domination}
\begin{theorem}[Free-channel domination]\label{thm:completion}
Let $\F$ be admissible and also satisfy
\ref{ax:permutation} and \ref{ax:marginal}. For every fixed
$0<\eps<1$ there is a constant $0\le K_\eps<\infty$ such that, for
every $n\ge1$ and $0<\delta\le1/16$, there are
$\mathcal S_n\in\F_n$ and $\mathcal L_n\in\CPTP(A^n\to B^n)$ with
\begin{align}
 \mathcal S_n&\cpge
 \beta_{\eps,n}\,2^{-g_\eps(n,\delta)}\mathcal L_n,
 \label{eq:completion-order}\\
 \dnorm{\mathcal L_n-\N_n}&\le\delta.
 \notag
\end{align}
Here $g_\eps(n,\delta)=K_\eps n^{2/3}\log_2((n+1)/\delta)$.
The constant depends only on the local dimensions, the minimum
eigenvalues of $\tau$ and $\omega$, and $\eps$.
It is independent of $n$, $\delta$, the target channel, and the
particular free family.
\end{theorem}

The usefulness of \cref{thm:completion} is that a free channel $\mathcal S_n$
contains an accurate approximation to the target as a component.
Indeed, set $p_n=\beta_{\eps,n}2^{-g_\eps(n,\delta)}$.
Since both $\mathcal S_n$ and $\mathcal L_n$ are trace preserving, the CP order
in \eqref{eq:completion-order} gives the convex decomposition
\[
 \mathcal S_n=p_n\mathcal L_n+(1-p_n)\mathcal Q_n,
 \qquad \mathcal Q_n\in\CPTP(A^n\to B^n).
\]
Here $0<p_n<1$ by \cref{lem:bounds}, and the residual channel
$\mathcal Q_n=(\mathcal S_n-p_n\mathcal L_n)/(1-p_n)$ need not be free.
Thus a free channel can contain a component that is $\delta$-close to $\N_n$ in
diamond norm, with a weight controlled by the testing coefficient.
For polynomially decreasing $\delta$, the loss $g_\eps(n,\delta)$ is $o(n)$, so
this component weight has the same exponential rate as $\beta_{\eps,n}$.
This connects the testing problem to an explicit CP domination of an
approximation to the target by a free channel.

\noindent\emph{Proof outline.}
For an isometric target, testing supplies a CP branch with positive
operator overlap in every input direction. Permutation averaging
and weighted discarding yield a compressed operator with a local
expansion. Polynomial approximations to its inverse and to the target
image projector correct the branch, while the additional family
conditions turn these corrections into free CP domination.
A free replacer restores trace preservation and pads the discarded
sites. Comparing testing against all auxiliary extensions reduces
an arbitrary target to the isometric case.
The full proof is given in \cref{app:quantitative-proof}.

\subsection{Finite-block smoothing bounds}
\begin{corollary}[Finite-block smoothing bounds]\label{cor:finite-AEP}
Assume the hypotheses of \cref{thm:completion}.
For every $0<\eps<1$, $0<\delta\le1/16$, and
$\eps+\delta/2<1$,
\begin{equation*}
 \begin{aligned}
 a_{\eps,n}+\log_2(1-\eps-\delta/2)
 &\le Z_n^\delta\\
 &\le a_{\eps,n}+K_\eps n^{2/3}
                      \log_2\frac{n+1}{\delta}.
 \end{aligned}
\end{equation*}
The lower bound requires only \cref{def:axioms} and holds for every
$0\le\delta<2(1-\eps)$.
\end{corollary}
\begin{proof}
The lower bound is \eqref{eq:AEP-lower}. For the upper bound,
\cref{thm:completion} supplies a channel within diamond distance
$\delta$ of $\N_n$ whose domination cost is at most
$a_{\eps,n}+K_\eps n^{2/3}\log_2((n+1)/\delta)$.
This channel is feasible in \eqref{eq:smoothed-resource-max}.
\end{proof}

\subsection{Comparison of testing tolerances}
\begin{lemma}\label{lem:error-independence}
Assume the hypotheses of \cref{thm:completion}.
For any $0<\eps_-\le\eps_+<1$, there is a finite constant $K$
such that
\begin{equation*}
 \sup_{\eps,\eps'\in[\eps_-,\eps_+]}
 |a_{\eps,n}-a_{\eps',n}|
 \le K\loss n\qquad(n\ge1).
\end{equation*}
In particular, the normalized difference tends to zero for every
two fixed errors in $(0,1)$.
\end{lemma}
\begin{proof}
For fixed $\eps,\eps'\in(0,1)$ choose
$\delta=\min\{1/16,(1-\eps')/2\}$ in \cref{thm:completion}.
Every tester feasible at error $\eps'$ accepts $\mathcal L_n$
with probability at least $1-\eps'-\delta>0$.
Since $\mathcal S_n$ is a free alternative, its CP domination gives
\begin{equation*}
 \beta_{\eps',n}\ge
 \beta_{\eps,n}2^{-g_\eps(n,\delta)}(1-\eps'-\delta).
\end{equation*}
Taking negative logarithms yields the upper bound
\[
 a_{\eps',n}-a_{\eps,n}
 \le g_\eps(n,\delta)-\log_2(1-\eps'-\delta)
 =O_{\eps,\eps'}(\loss n).
\]
Interchanging the errors gives the reverse bound.
For the uniform statement it suffices to apply this result to the
endpoints: $a_{\eps,n}$ is nondecreasing in $\eps$, so any difference
within the interval is bounded by $a_{\eps_+,n}-a_{\eps_-,n}$.
\end{proof}

\section{Resource families and exact benchmarks}\label{sec:applications}
\subsection{Admissible resource families}\label{sec:examples}

If $\dim A=1$, channels are states and \cref{thm:main} recovers
the generalized state Stein identity under \ref{ax:convex},
\ref{ax:tensor}, and \ref{ax:replacer}.
For the quantitative completion bound, the additional marginal axiom
becomes ordinary partial-trace closure.
The examples below satisfy both the basic conditions in
\cref{def:axioms} and the additional conditions
\ref{ax:permutation} and \ref{ax:marginal}. We verify these directly
and show why the witnessing input state for the quantitative bound
should be adapted to the resource family.
The exact examples give finite-block testing and smoothing
formulas. Table~\ref{tab:families} summarizes the channel families.

\begin{table}[t]
\centering
\small
\renewcommand{\arraystretch}{1.12}
\caption{Channel families satisfying the basic and additional
conditions in \cref{def:axioms,def:quantitative-axioms}.
Each row refers to all block channels satisfying the indicated
constraint. The witnessing states are those in the marginal and
replacer conditions.}
\label{tab:families}
\begin{tabularx}{\textwidth}{@{}Xll@{}}
\toprule
Channel family & Witnessing states & Verification\\
\midrule
Preserving $\tau^{\otimes n}\mapsto\omega^{\otimes n}$
 & Faithful $\tau,\omega$ & \Cref{prop:state-preserving}\\
Unital ($A=B=\mathbb C^d$)
 & $\tau=\omega=I/d$ & \Cref{prop:state-preserving}\\
Maximally incoherent or dephasing-covariant
 & $\tau=I/a$, $\omega=I/b$ & \Cref{prop:coherence}\\
Symmetry-covariant
 & Faithful invariant $\tau,\omega$ & \Cref{prop:covariant}\\
Entanglement-breaking or positive-partial-transpose
 & Any faithful $\tau,\omega$ & \Cref{prop:EB}\\
\bottomrule
\end{tabularx}
\end{table}

\begin{proposition}[Preservation of a faithful reference state]
\label{prop:state-preserving}
Fix faithful states $\tau\in\Dens(A)$ and $\omega\in\Dens(B)$.
The family
\begin{equation}
 \F_n^{\tau\to\omega}
 =\{\M\in\CPTP(A^n\to B^n):
                 \M(\tau^{\otimes n})=\omega^{\otimes n}\}
 \label{eq:state-preserving-family}
\end{equation}
is admissible and also satisfies \ref{ax:permutation} and
\ref{ax:marginal}, with the displayed $\tau$ and $\omega$.
For a replacer target $\N=\R_\rho$, one has the exact identities
\begin{equation}
 E_n=nD(\rho\Vert\omega),\qquad
 \beta_{\eps,n}
 =\beta_\eps^{\rm st}(\rho^{\otimes n}\Vert\omega^{\otimes n}),
 \label{eq:replacer-exact}
\end{equation}
where
$\beta_\eps^{\rm st}(\varrho\Vert \sigma)=
\min_{0\le Q\le I,\ \Tr Q \varrho\ge1-\eps}\Tr Q \sigma$.
In particular $R_\F(\R_\rho)=D(\rho\Vert\omega)$.
\end{proposition}
\begin{proof}
The preservation constraint is affine and closed inside the
compact convex set of channels.
The replacer $\R_\omega$ belongs to $\F_1^{\tau\to\omega}$.
Permutation and tensor closure follow because the prescribed
reference states are tensor powers.
For the one-site marginal $\M'$,
\[
 \M'(\tau^{\otimes(n-1)})
 =\Tr_{B_n}\M(\tau^{\otimes n})
 =\omega^{\otimes(n-1)}.
\]
This proves all axioms.

Using the input $\tau^{\otimes n}$ without a reference gives
outputs $\rho^{\otimes n}$ and $\omega^{\otimes n}$ for the target
and every alternative, respectively.
It follows that $E_n\ge nD(\rho\Vert\omega)$ and
$\beta_{\eps,n}\le\beta_\eps^{\rm st}$.
For the reverse bounds choose the free alternative
$\R_\omega^{\otimes n}$.
On any reference-assisted input state, the target and alternative
outputs are $\varrho_R\otimes\rho^{\otimes n}$ and
$\varrho_R\otimes\omega^{\otimes n}$.
Their relative entropy is $nD(\rho\Vert\omega)$.
Every joint effect is equivalent on these two outputs to the
single-system effect
\[
 Q_B=\Tr_R[(\varrho_R^{1/2}\otimes I)Q
                    (\varrho_R^{1/2}\otimes I)],\qquad 0\le Q_B\le I.
\]
Thus no tester improves on the displayed state-testing optimum.
This proves \eqref{eq:replacer-exact}.
\end{proof}

For finite-dimensional Hamiltonians at finite inverse temperature,
take $\tau$ and $\omega$ to be the corresponding Gibbs states.
Equation~\eqref{eq:state-preserving-family} then gives the full
Gibbs-preserving channel family \cite{FaistOppenheimRenner}.
When $A=B=\mathbb C^d$ and $\tau=\omega=I/d$, this is the
family of all unital block channels.
In particular a pure-state reset target has rate $\log_2d$
against unital alternatives.

The fixed-state marginal condition is strictly weaker than
closure under every auxiliary input state.
For $A=B$ and $\tau=\omega$, the two-site SWAP channel belongs
to $\F_2^{\tau\to\tau}$.
Inserting a state $\eta$ in the second input and tracing its
corresponding output yields $X\mapsto\Tr(X)\eta$, which preserves
$\tau$ only when $\eta=\tau$.
If $\tau\ne I/d$, even insertion of the maximally mixed state
fails this test.
Allowing a general faithful $\tau$ in the quantitative completion theorem therefore
enlarges its scope beyond a maximally-mixed-insertion formulation.

Coherence constraints give a second class of admissible families.
Fix product reference bases and let $\Delta_A,\Delta_B$ denote
complete dephasing in those bases.
The following are standard coherence-operation classes
\cite{ChitambarGour}.
At blocklength $n$, maximally incoherent operations (MIO) are channels mapping every
diagonal input state to a diagonal output state, equivalently
\[
 \Delta_B^{\otimes n}\circ\M\circ\Delta_A^{\otimes n}
 =\M\circ\Delta_A^{\otimes n}.
\]
Dephasing-covariant incoherent operations (DIO) are channels satisfying
$\Delta_B^{\otimes n}\circ\M=\M\circ\Delta_A^{\otimes n}$.

\begin{proposition}[Coherence-restricted channels]\label{prop:coherence}
The full MIO and DIO block-channel families are admissible and also
satisfy \ref{ax:permutation} and \ref{ax:marginal}, with
$\tau=I_A/a$ and $\omega=I_B/b$.
Consequently all results of
\cref{thm:main,thm:resource-AEP,thm:threshold} apply to these
families and arbitrary target channels.
\end{proposition}
\begin{proof}
Each defining equality is linear, so both sets are compact
and convex inside the set of channels.
Both contain the maximally mixed replacer.
Permutations commute with product dephasing.
Tensor products preserve the DIO equality.
They preserve MIO because every diagonal state on two blocks
is a convex combination of product basis-state projections;
each such projection is sent to a diagonal state.
Finally, let $\mathcal P(X)=X\otimes I_A/a$ and
$\mathcal T=\Tr_{B_n}$. These maps intertwine dephasing at the
two blocklengths:
\[
 \Delta_A^{\otimes n}\circ\mathcal P
 =\mathcal P\circ\Delta_A^{\otimes(n-1)},\qquad
 \Delta_B^{\otimes(n-1)}\circ\mathcal T
 =\mathcal T\circ\Delta_B^{\otimes n}.
\]
Combining these identities with the DIO equality for $\M$
shows that $\mathcal T\circ\M\circ\mathcal P$ remains DIO.
For MIO, a diagonal input on the remaining sites together with the
prepared maximally mixed state is diagonal; its output and output marginal
remain diagonal. This proves the marginal axiom as well.
\end{proof}

Covariance provides a further common structure, compatible with
intersections of resource constraints.

\begin{proposition}[Covariance and intersections]\label{prop:covariant}
Let $G$ be a group with finite-dimensional unitary representations
$U_A(g)$ and $U_B(g)$.
Suppose $\tau$ and $\omega$ are faithful invariant states.
The family of all channels obeying
\begin{equation*}
 \M\circ\Ad_{U_A(g)^{\otimes n}}
 =\Ad_{U_B(g)^{\otimes n}}\circ\M
 \quad(g\in G)
\end{equation*}
is admissible and also satisfies \ref{ax:permutation} and
\ref{ax:marginal}, with $\tau,\omega$.
Moreover, any intersection of families satisfying all five conditions
with the same witnessing states $\tau,\omega$ also satisfies all five.
\end{proposition}
\begin{proof}
The covariance constraints are closed and linear.
The invariant-state replacer is covariant and supplies
nonemptiness. Collective tensor representations commute with
site permutations, and tensor products preserve the covariance
equation.
Preparation of invariant $\tau$ intertwines the reduced and
full input actions; partial trace intertwines the corresponding
output actions. Thus the resulting marginal is covariant.

For an intersection, compactness, convexity, permutation
invariance, and tensor and marginal closure are inherited.
Nonemptiness is guaranteed by the common faithful replacer and
its tensor powers. All five axioms therefore hold.
\end{proof}

For Hamiltonians $H_A,H_B$, take the time-translation representations
$U_A(s)=e^{-isH_A}$ and $U_B(s)=e^{-isH_B}$ for $s\in\mathbb R$.
With Gibbs reference states, the proposition and \cref{prop:state-preserving}
cover the intersection of time-covariant and Gibbs-preserving
channels.

\subsection{Exact formulas and correlated alternatives}\label{sec:exact-examples}
A channel is entanglement breaking (EB) if its output is separable
from every reference for every input.
Equivalently its Choi operator is separable across reference
and output, or it has a measure-and-prepare representation
\cite{HorodeckiShorRuskai}.
A positive-partial-transpose (PPT) channel here means a channel whose Choi operator is
positive after transposition on the entire output system.
Equivalently $T_B\circ\M$ is CP, where $T_B$ is transposition.
At blocklength $n$, the relevant cut is $A^n:B^n$.

\begin{proposition}[Entanglement-breaking and PPT alternatives]\label{prop:EB}
The families of all EB block channels and all PPT block channels
are admissible and also satisfy \ref{ax:permutation} and
\ref{ax:marginal}, for any faithful input state $\tau$ and any
faithful output replacer state $\omega$.
For any isometry $J:A\to B$ and target $\N=\Ad_J$,
both families have the exact finite-block values
\begin{equation*}
 E_n=n\log_2a,\qquad
 \beta_{\eps,n}=(1-\eps)a^{-n},\qquad a=\dim A.
\end{equation*}
\end{proposition}
\begin{proof}
Normalized EB Choi operators form the intersection of the compact
convex separable-state set with the channel marginal constraint.
For PPT replace separability by positivity of the output partial
transpose, also a closed convex constraint.
Both families contain every replacer.
Tensor products and simultaneous permutations preserve their
respective Choi properties.

An EB channel remains EB after any input or output channel
composition: use reference separability after the input
composition and preservation of separability by a local output
channel. This proves its marginal closure.
For a PPT channel the reduced map $\M'$ satisfies
\[
 T_{B^{n-1}}\circ\M'
 =\Tr_{B_n}\circ T_{B^n}\circ\M
             \circ(X\mapsto X\otimes\tau).
\]
The right-hand side is CP because preparation and partial trace
are CP and $T_{B^n}\circ\M$ is CP by the PPT assumption.
The displayed identity uses
$\Tr_{B_n}\circ T_{B^n}=T_{B^{n-1}}\circ\Tr_{B_n}$:
transposition on the discarded factor leaves its trace unchanged,
while transposition on the remaining factors passes through the
partial trace. This proves the PPT marginal axiom.
Every reference-assisted PPT-channel output is PPT because
$T_{B^n}\circ\M$ is CP.
Every EB output is PPT as well.

Put $D=a^n$, $J_n=J^{\otimes n}$, and use the normalized
maximally entangled input
$\ket{\Phi_D}=D^{-1/2}\sum_{j=1}^D\ket j_R\ket j_{A^n}$.
The target output is the pure state
$\rho_J=\ket{\phi_J}\bra{\phi_J}$ with
$\ket{\phi_J}=(I_R\otimes J_n)\ket{\Phi_D}$.
All its nonzero Schmidt coefficients equal $D^{-1/2}$.
Write $X^{T_B}$ for partial transposition on the entire $B^n$
output in this calculation. In Schmidt bases,
$\rho_J^{T_B}$ is $1/D$ times the swap operator on a
$D$-by-$D$ subspace and zero on its orthogonal complement.
Consequently $\|\rho_J^{T_B}\|_\infty=1/D$.
For every alternative output $\sigma$ from either family,
$\sigma^{T_B}$ is positive with trace one. Therefore
\begin{equation}
 \bra{\phi_J}\sigma\ket{\phi_J}
 =\Tr\rho_J^{T_B}\sigma^{T_B}\le1/D .
 \label{eq:sep-overlap}
\end{equation}
The effect $(1-\eps)\ket{\phi_J}\bra{\phi_J}$ proves
$\beta_{\eps,n}\le(1-\eps)/D$.
Binary measurement onto $\rho_J$ also gives
$D(\rho_J\Vert\sigma)\ge\log_2D$, so $E_n\ge\log_2D$.

For both reverse bounds, let $\Delta$ be complete dephasing
in any basis of $A^n$ and take the EB alternative
$\mathcal M_*=\Ad_{J_n}\circ\Delta$.
Its Kraus operators $K_j=J_n\ket j\bra j$ satisfy
$\sum_jK_j=J_n$.
Cauchy--Schwarz on their Choi vectors gives
\begin{equation}
 \sum_j\ket{K_j}\bra{K_j}
 \ge D^{-1}\ket{J_n}\bra{J_n},
 \qquad \mathcal M_*\cpge D^{-1}\Ad_{J_n}.
 \label{eq:pinching-CP-benchmark}
\end{equation}
Thus every feasible tester has alternative acceptance at least
$(1-\eps)/D$, and operator monotonicity of logarithm yields
$D_{\ch}(\Ad_{J_n}\Vert\mathcal M_*)\le\log_2D$.
Since EB channels are PPT, the same alternative proves both
upper bounds on $E_n$ and both lower bounds on $\beta_{\eps,n}$.
\end{proof}

\begin{corollary}[Exact smoothed isometric benchmark]
\label{cor:exact-smoothed-isometry}
For either family in \cref{prop:EB}, an isometric target, and
every $0\le\delta\le2$,
\begin{equation}
 Z_n^\delta=\log_2\max\{1,a^n(1-\delta/2)\}.
 \label{eq:exact-smoothed-isometry}
\end{equation}
The factor $1/2$ here corresponds to our use of the full diamond
norm in the smoothing constraint.
\end{corollary}
\begin{proof}
Put $D=a^n$. If $D=1$, every channel with this input is a
replacer and the assertion is immediate. Assume $D>1$.
Let $\mathcal L$ be a feasible smoothing channel and let $\sigma$
be its output on the maximally entangled input state used above.
Since $\sigma$ and $\rho_J$ are states,
\[
 \bra{\phi_J}\sigma\ket{\phi_J}\ge1-\tfrac12\trnorm{\sigma-\rho_J}
 \ge1-\delta/2.
\]
If $\mathcal L\cple\lambda\mathcal S$ with $\mathcal S$
free, \eqref{eq:sep-overlap} gives
$1-\delta/2\le\lambda/D$. Together with $\lambda\ge1$,
this proves the lower bound in \eqref{eq:exact-smoothed-isometry}.

For the upper bound use the same
$\mathcal M_*=\Ad_{J_n}\circ\Delta$ as in \cref{prop:EB}.
In a basis indexed by $j=0,\ldots,D-1$, let
$U\ket j=e^{2\pi i j/D}\ket j$.
Then $\Delta=D^{-1}\sum_{l=0}^{D-1}\Ad_{U^l}$ and
\[
 \dnorm{\mathcal M_*-\Ad_{J_n}}
 \le \dnorm{\Delta-\id}\le 2(1-1/D).
\]
Indeed, the $l=0$ term is $D^{-1}\id$, and the other terms
form $(1-1/D)$ times a channel.
Set
\[
 s=\min\{1,\delta/[2(1-1/D)]\},\qquad
 \mathcal L_s=(1-s)\Ad_{J_n}+s\mathcal M_*.
\]
This is CPTP and within diamond distance $\delta$ of the target.
Equation~\eqref{eq:pinching-CP-benchmark} gives
\[
 \begin{gathered}
 \mathcal L_s\cple[(1-s)D+s]\mathcal M_*,\\
 (1-s)D+s=\max\{1,D(1-\delta/2)\}.
 \end{gathered}
\]
The same EB dominator is available in both families, completing
the proof.
\end{proof}

The preceding exact formulas describe exponential discrimination.
The following construction exhibits how persistent correlations
can instead produce a zero rate for a target outside the one-use
free set.

\begin{example}[Correlated alternatives with zero Stein rate]\label{sec:nonfaithfulness}
Work in the state specialization $A=\mathbb C$.
Choose a faithful state $\omega$, a unit vector $\ket\phi$ with
$\rho=\ket\phi\bra\phi\ne\omega$, and $0<c<1$.
On a block of size $k$ put
$\zeta_k=c\rho^{\otimes k}+(1-c)\omega^{\otimes k}$.
Let $\F_n$ be the convex hull of the states obtained by partitioning
the $n$ labeled sites into blocks and assigning to each block of
size $k$ either $\zeta_k$ or $\omega^{\otimes k}$.
For fixed $n$ this is a finite convex hull and hence compact.
All partitions ensure permutation invariance; joining partitions
ensures tensor closure. Tracing a site from a block replaces
$\zeta_k$ by $\zeta_{k-1}$ or removes a singleton, and preserves
the analogous product state. Thus marginal closure holds.
The faithful product state is included, so the family satisfies
both the basic and the additional conditions.
Writing $\operatorname{conv}$ for the convex hull, the one-use set is
$\F_1=\operatorname{conv}\{\omega,c\rho+(1-c)\omega\}$,
so $\rho\notin\F_1$. At blocklength $n$,
$\zeta_n\ge c\rho^{\otimes n}$ gives, identifying $\rho$ with
its preparation channel,
\[
 0\le E_n\le\log_2(1/c),\qquad
 \beta_{\eps,n}\ge c(1-\eps),\qquad R_\F(\rho)=0 .
\]
The correlated alternative $\zeta_n$ contains a target component
of fixed weight $c$ at every blocklength. This keeps the worst-case
type-II error bounded away from zero and explains the vanishing rate.
\end{example}

\section{Discussion}\label{sec:discussion}
For parallel discrimination against correlated alternative channels,
the regularized Umegaki channel relative entropy of resource is the
fixed-error Stein exponent, the CPTP-smoothed max-relative entropy
rate, and the threshold above which target acceptance decays
exponentially. Compactness and convexity, tensor closure, and a
faithful free replacer suffice for all these asymptotic statements.

Uniform auxiliary map approximation provides the operator control
needed to prove this correspondence. Iterative reduction of the
domination rate and amplification of a fixed accurate block produce
exponential diamond accuracy. These approximations establish the
ordinary testing limits directly, and an entropy comparison identifies
their common value. The same construction gives equipartition whenever
$0<\delta_n\le\bar\delta<2$ and $\log_2(1/\delta_n)=o(n)$, as well
as an exponential strong converse.

The quantitative completion theorem controls a different aspect of
the problem: its bound applies at a specified blocklength and accuracy
with a constant uniform over the stated class of targets and families.
Permutation invariance and fixed-input marginal closure permit weighted
discarding and local corrections, yielding an explicit overhead
$O_\varepsilon(n^{2/3}\log_2((n+1)/\delta))$ relative to the testing
coefficient. The faithful input state can be chosen to respect a
preferred reference state of the resource family.
This construction also gives an alternative proof of the Stein
identity through near subadditivity.

The applications include state preservation, coherence restrictions,
covariance, entanglement breaking, and positive partial transpose.
The isometric benchmarks resolve testing and smoothing exactly at
every blocklength. The zero-rate construction shows how persistent
correlations can suppress exponential discrimination even when the
target is outside the one-use free set.

The access model is parallel, with arbitrary correlated block
alternatives. Extensions to sequential access require a corresponding
causal model for those alternatives.
Further questions include the optimal strong-converse exponent,
equipartition at fixed positive exponential smoothing rates, and
resource-conversion rates under specified classes of superchannels.

\appendix
\crefalias{section}{appendix}
\crefalias{subsection}{appendix}
\crefalias{subsubsection}{appendix}
\section{Proof of quantitative completion}\label{app:quantitative-proof}

We prove \cref{thm:completion} under \cref{def:axioms} and the
additional conditions \ref{ax:permutation} and \ref{ax:marginal}.
We first extract a branch with uniform positive overlap and compare
testing against auxiliary extensions. Local operator estimates
then give quantitative completion for isometric targets, followed
by the reduction for general targets.

\subsection{Branch extraction and auxiliary extensions}\label{app:testing-lifts}
A completely positive branch of a channel $\M$ is a CP map
$\mathcal B$ for which $\M-\mathcal B$ is CP. Positivity of both
maps and trace preservation of $\M$ imply
$0\le\mathcal B^{\dagger}(I)\le\M^{\dagger}(I)=I$ and
$0\le(\M-\mathcal B)^{\dagger}(I)\le I$.
Thus both maps are trace nonincreasing; together
they can be viewed as the two outcomes of an instrument.
For $\mathcal B=p\Ad_V$, the probability of the branch on an
input state $\rho$ is $p\Tr V\rho V^{\dagger}$, which need not equal $p$.
In particular, the branch coefficient and its input-dependent
probability must be distinguished.

We first extract a branch with positive overlap from a testing
lower bound. We then compare discrimination of a channel with
discrimination of a Stinespring isometry against all extensions
of the original alternatives to an auxiliary output system.
An isometric target produces a pure output on a purified input,
making the branch argument available.
Access to its auxiliary output can improve discrimination, however, so
an explicit comparison is needed before using that argument for a
noisy target.

\subsubsection{A branch with uniform positive overlap}
\begin{lemma}\label{lem:branch}
Let $J:H\to K$ be an isometry and let $\M:\mathrm L(H)\to\mathrm L(K)$ be CPTP.
Fix $0<\delta<1$, and suppose
$\beta=b_\delta(\Ad_J,\M)>0$.
Set $r=\sqrt{1-\delta}$ and $c=(1-r)/(1+r)$.
Then there is $V:H\to K$ such that
\begin{equation}
 \opnorm V\le1,\qquad
 \Rea(J^{\dagger}V)\ge cI,\qquad
 \M\cpge\beta\Ad_V.
 \label{eq:extracted-branch}
\end{equation}
\end{lemma}
\begin{proof}
We construct $V$ by finding a Kraus combination $A$ close in
operator norm to a positive multiple of $J$, and then rescaling $A$.
The Kraus representation will give CP domination, while closeness
to $J$ will give the norm and overlap bounds in
\eqref{eq:extracted-branch}. To obtain this approximation, we show
that its failure would produce a feasible test with alternative
acceptance strictly below $\beta$.

Choose Kraus operators $K_i$ for $\M$ and define
\[
 \mathcal E=\left\{\sum_i z_iK_i:\sum_i|z_i|^2\le1\right\}.
\]
This set is compact and convex because it is the linear image of
the Euclidean unit ball of Kraus coefficients. We first check
that every $A\in\mathcal E$ gives a CP map dominated by $\M$.
Fix $A=\sum_i z_iK_i\in\mathcal E$, so that $\sum_i|z_i|^2\le1$.
By linearity of vectorization, $\ket A=\sum_i z_i\ket{K_i}$.
The Kraus representation of $\M$ and the definition of $\Ad_A$ give
\[
 C_\M=\sum_i\ket{K_i}\bra{K_i},\qquad
 C_{\Ad_A}=\ket A\bra A.
\]
For any vector $\ket v\in H\otimes K$, Cauchy--Schwarz therefore yields
\begin{align*}
 |\langle v\mid A\rangle|^2
 &=\left|\sum_i z_i\langle v\mid K_i\rangle\right|^2\\
 &\le\Bigl(\sum_i|z_i|^2\Bigr)
       \sum_i|\langle v\mid K_i\rangle|^2\\
 &\le\sum_i|\langle v\mid K_i\rangle|^2
 =\bra v C_\M\ket v.
\end{align*}
Thus $\bra v(C_\M-C_{\Ad_A})\ket v\ge0$ for every $\ket v$,
which means $C_\M-C_{\Ad_A}\ge0$. By the Choi criterion,
$\M-\Ad_A$ is completely positive, or equivalently
$\M\cpge\Ad_A$.

We now seek $A\in\mathcal E$ such that
\[
 \opnorm{uJ-A}\le ru,\qquad u=\frac{\sqrt\beta}{1-r}>0.
\]
This is the approximation needed for the final rescaling.
To prove that it exists, we express the distance from $uJ$ to
$\mathcal E$ as an optimization over operators $X$ that can be
converted into tests.
For $X\in\mathrm L(H,K)$, the support function of $\mathcal E$
with respect to the real trace pairing is defined by
\[
 h_{\mathcal E}(X)=\max_{A\in\mathcal E}\Rea\Tr X^{\dagger}A.
\]
It is the largest value of the real linear functional
$A\mapsto\Rea\Tr X^{\dagger}A$ on $\mathcal E$.
Writing $c_i=\Tr X^{\dagger}K_i$, we obtain
\[
 h_{\mathcal E}(X)
 =\max_{\sum_i|z_i|^2\le1}\Rea\sum_i z_ic_i
 =\left(\sum_i|c_i|^2\right)^{1/2}.
\]
Indeed, Cauchy--Schwarz gives the upper bound, and for $c\ne0$
it is attained by $z_i=\overline{c_i}/\|c\|_2$; when $c=0$,
both sides vanish.

To express the operator-norm distance from $uJ$ to $\mathcal E$,
use duality between the operator and trace norms:
\[
 \opnorm Y=\max_{\trnorm X\le1}\Rea\Tr X^{\dagger}Y.
\]
Applying this identity to $Y=uJ-A$ gives
\[
 \min_{A\in\mathcal E}\opnorm{uJ-A}
 =
 \min_{A\in\mathcal E}\max_{\trnorm X\le1}
       \Rea\Tr X^{\dagger}(uJ-A).
\]
Both feasible sets are compact and convex. For each fixed $u$,
the objective is continuous and affine in each variable, so minimax
exchanges the extrema. The inner minimum over $A$ then equals
$u\Rea\Tr X^{\dagger}J-h_{\mathcal E}(X)$. Substituting the
support function computed above yields
\begin{equation}
 \begin{aligned}
 &\min_{A\in\mathcal E}\opnorm{uJ-A}\\
 &=\max_{\trnorm X\le1}
 \left\{u\Rea\Tr X^{\dagger}J-
       \left(\sum_i|\Tr X^{\dagger}K_i|^2\right)^{1/2}\right\}.
 \end{aligned}
 \label{eq:ellipsoid-distance}
\end{equation}

Suppose, for a contradiction, that
$\min_{A\in\mathcal E}\opnorm{uJ-A}>ru$.
By \eqref{eq:ellipsoid-distance}, there is an operator $X$ with
$\trnorm X\le1$ such that
\[
 u\Rea\Tr X^{\dagger}J
 -\left(\sum_i|\Tr X^{\dagger}K_i|^2\right)^{1/2}>ru.
\]
Denote the left-hand side by $F(X)$, so that $F(X)>ru>0$.
In particular, $X\ne0$, and hence $0<\trnorm X\le1$.
Both terms in $F(X)$ scale by $s$ when $X$ is replaced by $sX$,
for $s>0$. Therefore
\[
 F\!\left(\frac{X}{\trnorm X}\right)
 =\frac{F(X)}{\trnorm X}\ge F(X)>ru,
\]
where the middle inequality uses $F(X)>0$ and $\trnorm X\le1$.
The normalized operator thus still satisfies the strict inequality,
and we may replace $X$ by it and assume $\trnorm X=1$.
Define
\[
 \alpha=\Rea\Tr X^{\dagger}J,\qquad
 q_X=\sum_i|\Tr X^{\dagger}K_i|^2.
\]
Rearranging the strict inequality gives
$\sqrt{q_X}<u(\alpha-r)$. Since its left-hand side is
nonnegative and $u>0$, we obtain $\alpha>r$.
On the other hand, trace-norm duality and the isometry property give
\[
 \alpha\le|\Tr X^{\dagger}J|
 \le\trnorm X\opnorm J=1.
\]
Consequently,
\[
 \sqrt{q_X}<u(\alpha-r)\le u(1-r)=\sqrt\beta,
\]
so $q_X<\beta$. We now use $X$ to construct a test whose target
acceptance exceeds $r^2=1-\delta$ and whose alternative acceptance
is $q_X$. This will contradict the definition of $\beta$.
Write a singular-value decomposition
$X=\sum_j x_j\ket{u_j}\bra{v_j}$ with $\sum_jx_j=1$.
The normalized input and measurement vectors
\[
 \ket\psi=\sum_j\sqrt{x_j}\ket j_R\ket{v_j}_H,\qquad
 \ket\phi=\sum_j\sqrt{x_j}\ket j_R\ket{u_j}_K
\]
define a feasible test. Indeed, its effect $Q=\ket\phi\bra\phi$
has target acceptance
\[
 \Tr Q(\id_R\otimes\Ad_J)(\ket\psi\bra\psi)
 =|\Tr X^{\dagger}J|^2\ge\alpha^2>1-\delta.
\]
For the alternative channel, the Kraus representation gives
\begin{align*}
 &\Tr Q(\id_R\otimes\M)(\ket\psi\bra\psi)\\
 &\quad=\sum_i|\bra\phi(I_R\otimes K_i)\ket\psi|^2\\
 &\quad=\sum_i|\Tr X^{\dagger}K_i|^2=q_X<\beta.
\end{align*}
This contradicts $\beta=b_\delta(\Ad_J,\M)$, the minimum
alternative acceptance among tests with target acceptance at least
$1-\delta$. Hence the distance is at most $ru$. Compactness of
$\mathcal E$ ensures that some $A\in\mathcal E$ attains it.

It remains to obtain the three bounds in \eqref{eq:extracted-branch}.
Write $A=uJ+D$, where $\opnorm D\le ru$, and set
$V=A/[u(1+r)]$. The triangle inequality gives
\[
 \opnorm V\le\frac{u\opnorm J+\opnorm D}{u(1+r)}\le1.
\]
Since $J^{\dagger}J=I$ and
$\Rea(J^{\dagger}D)\ge-\opnorm{J^{\dagger}D}I\ge-ruI$,
we also have
\[
 \Rea(J^{\dagger}V)
 =\frac{uI+\Rea(J^{\dagger}D)}{u(1+r)}
 \ge\frac{1-r}{1+r}I=cI.
\]
Finally, the CP domination established for every member of
$\mathcal E$ yields
\[
 \M\cpge\Ad_A
 =u^2(1+r)^2\Ad_V
 =\frac{\beta}{c^2}\Ad_V\cpge\beta\Ad_V.
\]
The last inequality uses $0<c<1$ and complete positivity of
$\Ad_V$, completing the proof.
\end{proof}

\subsubsection{Comparison with auxiliary extensions}
\begin{lemma}
\label{lem:extension}
Let $\N:\mathrm L(A)\to\mathrm L(B)$ be CPTP, let
$\mathcal F\subseteq\CPTP(A\to B)$ be nonempty, compact, and convex,
and fix a Stinespring isometry $V:A\to B\otimes E$.
Define
\[
 \mathcal G=\{\mathcal L\in\CPTP(A\to B\otimes E):
                    \Tr_E\circ\mathcal L\in\mathcal F\}.
\]
Write $\beta_\eps$ for worst-case testing of $\N$ against $\mathcal F$
and $\gamma_\delta$ for testing of $\Ad_V$ against $\mathcal G$.
For $0<\eps<1$, set
\begin{equation*}
 s=\sqrt{1-\eps},\quad r=(1+s)/2,\quad
 \delta=1-r^2,\quad t=(1-s)^2/4.
\end{equation*}
Then $0<\delta<1$, $t>0$, and
\begin{equation}
 t\beta_\eps\le\gamma_\delta,\qquad
 \gamma_\eps\le\beta_\eps.
 \label{eq:extension-comparison}
\end{equation}
The parameters are independent of the dimensions.
\end{lemma}
\begin{proof}
The set $\mathcal G$ is nonempty: append a fixed state on $E$ to
any channel in $\mathcal F$. It is compact and convex, being the
inverse image of $\mathcal F$ under the continuous linear partial-trace
map, intersected with the compact convex set of channels.
The second inequality follows by lifting any original effect $Q$
to $Q\otimes I_E$, preserving all relevant acceptance probabilities.

For the first inequality, we must rule out an excessive advantage
from access to $E$. A lifted effect induces, through a dual
optimization and spectral thresholding, a feasible effect on the
original output. The constants in this construction depend only
on the acceptance tolerance, not on $\dim E$.

Use \cref{lem:testing-minimax} to choose
a hardest $\M\in\mathcal F$ with
$b_\eps(\N,\M)=\beta_\eps$.
Fix any feasible lifted tester. Purify its input and extend its
effect by the identity on the additional reference, so that its
input state is represented by a unit vector $\ket\psi$. Set
\[
 \begin{gathered}
 \ket\phi=(I_R\otimes V)\ket\psi,\qquad
 \rho=\Tr_E\ket\phi\bra\phi,\\
 \sigma=(\id\otimes\M)(\ket\psi\bra\psi).
 \end{gathered}
\]
For this fixed input state, the outputs of channels extending $\M$
are exactly the states $\tau_{RBE}$ with $\Tr_E\tau=\sigma$.
To verify the nontrivial direction, apply a Stinespring dilation
of $\M$ to $\ket\psi$, obtaining a purification of $\sigma$
with auxiliary system $Z$. Purify the desired $\tau$ further.
The Schmidt decompositions provide an isometry from the support
of the first purification's $Z$ marginal into the second purifying
space. Tracing out the extra register gives a channel on that
support from $Z$ to $E$; extend it to the whole $Z$ by sending the
orthogonal complement to a fixed state. Acting with this channel
on the dilation of $\M$ gives an extension of $\M$ on \emph{all}
inputs and produces the required $\tau$ on $\ket\psi$.

Let $Q$ be the lifted effect. Its largest acceptance over these
extensions is the following semidefinite program (SDP), with dual
variable $H=H^{\dagger}$ acting on $R\otimes B$:
\begin{align}
 c_Q&=\max_{\tau\ge0,\ \Tr_E\tau=\sigma}\Tr Q\tau\notag\\
 &=\inf_{\substack{H=H^{\dagger}\\H\otimes I_E\ge Q}}\Tr H\sigma.
 \label{eq:extension-sdp}
\end{align}
The primal is nonempty and compact, and $H=2I$ is strictly dual
feasible; strong duality applies. We use a minimizing sequence
of dual feasible operators, which also covers singular $\sigma$.
Every feasible $H$ is positive, since $H\otimes I_E\ge Q\ge0$.

Fix such $H$ and let $P=\mathbf1_{[t,\infty)}(H)$ be its spectral
projection onto eigenvalues at least $t$. Put $p=\Tr P\rho$.
Split $\ket\phi$ into its $P$ and $I-P$ components.
Using $Q\le I$ and $Q\le H\otimes I_E$ gives
\begin{align*}
 \sqrt{\langle\phi|Q|\phi\rangle}
 &\le\|Q^{1/2}(P\otimes I_E)\ket\phi\|\\
 &\quad+\|Q^{1/2}((I-P)\otimes I_E)\ket\phi\|\\
 &\le\sqrt p+\sqrt t.
\end{align*}
The first component has norm at most $\sqrt p$ because $Q\le I$.
On the second component, $H$ has eigenvalues below $t$, so
$Q\le H\otimes I_E$ bounds its squared norm by $t(1-p)\le t$.
The target acceptance is at least $1-\delta=r^2$.
Since $r-\sqrt t=s$, we obtain $p\ge s^2=1-\eps$.
Thus $P$, with the same original input state, is a feasible test of
$\N$ against $\M$.
Its alternative acceptance is bounded by
$\Tr P\sigma\le\Tr H\sigma/t$, whence
\[
 \beta_\eps=b_\eps(\N,\M)\le\Tr H\sigma/t.
\]
Taking the infimum in \eqref{eq:extension-sdp} yields
$c_Q\ge t\beta_\eps$. The lifted tester's worst-case alternative
acceptance is at least $c_Q$. Taking the infimum over lifted testers
proves the first inequality in \eqref{eq:extension-comparison}.
Here the auxiliary projector estimates the single-alternative
testing value at the fixed hardest $\M$. The original lifted
tester remains common to all extensions.
\end{proof}

\subsection{Local approximation and free domination}\label{sec:local}\label{app:completion-tools}

Three lemmas control the local corrections used in completion.
After discarding a sublinear number of sites, the first lemma
approximates a permutation-invariant operator by terms supported
on few sites. The second converts an expansion of this kind into
a free CP dominator. The third constructs a compatible approximation
to the projector onto an isometry's product image.

An operator on a tensor product is \emph{supported on} a set of sites $S$
if it equals an operator on $S$ tensored with the identity elsewhere.
A finite local-support expansion is an expression
$O=\sum_i c_iF_i$ with $\opnorm{F_i}\le1$.
Its coefficient cost is $\sum_i|c_i|$.
We track both the maximum support size and this cost.
In this section $\supp F_i$ denotes a chosen set of sites supporting
$F_i$, and $|\supp F_i|$ is its cardinality.
The expansions allow arbitrary complex coefficients and
supported contractions. In particular, the $F_i$ need not be
Pauli operators or elements of any fixed operator basis. Their site
support is distinct from the range support of a positive operator.
For products, support sizes add as an upper bound and coefficient
costs multiply: if $O_1=\sum_i c_iF_i$ and $O_2=\sum_jd_jG_j$,
then $F_iG_j$ is a contraction on the union of the two supports
and the resulting cost is at most $(\sum_i|c_i|)(\sum_j|d_j|)$.

\subsubsection{Weighted discarding}
Permutation symmetry controls the algebra in which an operator
lives, but does not itself make that operator local. Weighted
discarding provides the additional regularization. The following
lemma quantifies both the approximation error and the cost of a
local expansion; both estimates are needed for free domination.

\begin{lemma}
\label{lem:local-approximation}
Fix a faithful $\tau\in\Dens(\mathbb C^d)$ and put
$t=\lambda_{\min}(\tau)>0$.
There are constants $n_0$ and $C_{d,t}$, depending only on $d$ and $t$,
with the following property.
For $n\ge n_0$ and $0<\xi\le1/16$, set
\begin{equation*}
 \ell_n=\ln\frac{n+1}{\xi},\qquad
 m=\lceil n^{2/3}\ell_n\rceil,\qquad k=n-m,
\end{equation*}
and assume $m\le n/2$.
For every permutation-invariant contraction $W$ on
$(\mathbb C^d)^{\otimes n}$, define
\[
 X=\mathbb E_{\tau,m}(W)
 :=\Tr_m[(I\otimes(\tau^{1/2})^{\otimes m})
                W(I\otimes(\tau^{1/2})^{\otimes m})].
\]
Here $\Tr_m$ traces out the last $m$ sites, leaving the first $k$.
There is a finite expansion $P=\sum_i c_iF_i$ on $k$ sites such that
\begin{align}
 \opnorm{P-X}&\le\xi,\label{eq:local-error}\\
 |\supp F_i|&\le C_{d,t}n^{2/3},\notag\\
 \sum_i|c_i|&\le\exp(C_{d,t}n^{2/3}),
 \qquad \opnorm{F_i}\le1.
 \label{eq:local-coefficients}
\end{align}
The constants are uniform in $\xi$ and $W$.
The map $\mathbb E_{\tau,m}$ is unital and completely positive.
\end{lemma}
\begin{proof}
\emph{Step 1: A controlled integral representation.}
The last assertion follows directly from the displayed sandwich
formula and $\Tr\tau^{\otimes m}=1$.
For $d=1$ take $P=X=W$, supported on the empty set.
For $d>1$, let $U(d)$ denote the unitary group on $\mathbb C^d$.
Index partitions of $n$ by $\lambda\vdash n$, and write
$\ell(\lambda)$ for the number of their nonzero parts.
Schur--Weyl duality \cite{FultonHarris} gives
\[
 (\mathbb C^d)^{\otimes n}
 =\bigoplus_{\lambda\vdash n,\ \ell(\lambda)\le d}
 Q_\lambda\otimes P_\lambda,\qquad
 W=\bigoplus_\lambda W_\lambda\otimes I_{P_\lambda}.
\]
Here $Q_\lambda$ carries an irreducible representation of $U(d)$, and
$P_\lambda$ carries one of the permutation group. The two commuting
actions, denoted by $D_\lambda(g)$ on $Q_\lambda$ and
$\Pi_\lambda(\pi)$ on $P_\lambda$, have the forms
\[
 g^{\otimes n}=\bigoplus_\lambda D_\lambda(g)\otimes I_{P_\lambda},
 \qquad U_\pi=\bigoplus_\lambda I_{Q_\lambda}\otimes\Pi_\lambda(\pi).
\]
Commutation of $W$ with every $U_\pi$ gives the displayed block
form of $W$. This is the precise commutant consequence of
Schur--Weyl duality used here.
Put $r_\lambda=\dim Q_\lambda$.
Schur orthogonality, with normalized Haar measure, yields
\begin{align}
 W&=\int_{U(d)}f_W(g)g^{\otimes n}\,dg,\label{eq:fourier-W}\\
 f_W(g)&=\sum_\lambda r_\lambda
              \Tr[W_\lambda D_\lambda(g)^{\dagger}].\notag
\end{align}
Indeed, the integral of
$\Tr[Z D_\lambda(g)^{\dagger}]D_\mu(g)$ is zero for $\lambda\ne\mu$
and equals $Z/r_\lambda$ otherwise, by entrywise orthogonality.
Since $\opnorm{W_\lambda}\le1$,
\begin{equation*}
 |f_W(g)|\le v_n:=\sum_\lambda r_\lambda^2
 \le\binom{n+d^2-1}{d^2-1}
 \le(n+d^2+1)^{d^2-1}.
\end{equation*}
The sum is the dimension of the permutation commutant.
To bound it, expand operators in tensor products of the $d^2$
matrix units. Invariance under site permutations makes each
coefficient constant on an orbit. Each orbit is determined by
$d^2$ multiplicities summing to $n$, so there are
$\binom{n+d^2-1}{d^2-1}$ such orbits.

\emph{Step 2: Damping under weighted discarding.}
The function $f_W$ may be complex valued; this representation
is a linear integral, not a probability mixture. Its useful feature
is the polynomial bound on its absolute coefficient mass. On a
unitary tensor power, weighted discarding acts as
\[
 \mathbb E_{\tau,m}(g^{\otimes n})
       =[\Tr(\tau g)]^m g^{\otimes k}.
\]
Each discarded site contributes one scalar factor. Applying
$\mathbb E_{\tau,m}$ to \eqref{eq:fourier-W} therefore gives
\begin{equation}
 X=\int f_W(g)\,[\Tr(\tau g)]^m g^{\otimes k}\,dg .
 \label{eq:discard-Fourier}
\end{equation}
This formula explains the regularization: if
$|\Tr(\tau g)|$ is appreciably below one, its $m$th power damps
the contribution. If it is close to one, faithfulness of $\tau$
forces $g$ close to a scalar unitary, as quantified below.
Choose
\[
 \begin{gathered}
 h=(d^2+2)\ln(n+d^2+1)+\ln(4/\xi),\\
 u=\sqrt{\frac{2h}{tm}},\qquad
 R=\lceil e^2ku+2h\rceil.
 \end{gathered}
\]
Uniformly in the permitted $\xi$,
$h=O_d(\ell_n)$ and $u=O_{d,t}(n^{-1/3})$.
The active-regime assumption $m\le n/2$ also gives
$\ell_n\le n^{1/3}/2$, hence
$R=O_{d,t}(n^{2/3})$.
Increasing $n_0$ ensures $u<1$ and $R<k$.

\emph{Step 3: Concentration near scalar unitaries.}
Divide the integral into two regions, according to this damping.
Call $g$ good when $|\Tr(\tau g)|\ge1-tu^2/2$, and put
$z=\Tr(\tau g)/|\Tr(\tau g)|$ on that set.
Positivity of $\tau-tI$ implies
\begin{align*}
 t\opnorm{g-zI}^2
 &\le t\Tr[(g-zI)^{\dagger}(g-zI)]\\
 &\le\Tr[\tau(g-zI)^{\dagger}(g-zI)]\\
 &=2(1-|\Tr(\tau g)|)\le tu^2.
\end{align*}
For bad $g$, $|\Tr(\tau g)|^m\le e^{-h}$.
Its contribution to \eqref{eq:discard-Fourier} therefore has norm
at most $v_ne^{-h}$.
The estimate applies to arbitrary faithful $\tau$ and unitary $g$,
including the noncommuting case.

\emph{Step 4: Truncation and a finite local expansion.}
On the good set write $g=z(I+H)$, where $\opnorm H\le u$,
and expand $g^{\otimes k}=z^k\sum_{S\subseteq[k]}H_S$.
Here $H_S$ acts as $H$ on $S$ and as the identity elsewhere.
Truncate to $|S|\le R$ and integrate to define $P$.
The omitted tail obeys
\begin{equation*}
 \sum_{r>R}\binom kr u^r
 \le\sum_{r>R}\left(\frac{eku}{r}\right)^r
 \le 2e^{-R},
\end{equation*}
because $R\ge e^2ku$.
As $R\ge2h$ and $v_ne^{-h}\le\xi/4$, it follows that
\[
 \opnorm{P-X}\le v_ne^{-h}+2v_ne^{-R}
 \le \xi/4+2v_ne^{-2h}\le\xi .
\]
Each $H_S/u^{|S|}$ is a supported contraction, and the integral
expansion has absolute coefficient mass at most
\[
 v_n\sum_{r\le R}\binom kr u^r
 \le v_ne^{ku}\le\exp[O_{d,t}(n^{2/3})].
\]
Absorb the complex phases into the contractions, normalize by this
upper bound, and add a zero contraction if necessary.
The resulting barycenter lies in the convex hull of the compact
union of contraction balls supported on at most $R$ sites.
Carath\'eodory's theorem in a finite-dimensional real vector space
gives a finite convex representation without increasing its cost.
This proves \eqref{eq:local-error}--\eqref{eq:local-coefficients}.
\end{proof}

\subsubsection{Free multipliers}
The next lemma turns a local-support expansion into CP domination
by a member of the free family. The correction operator itself
need not represent a free operation. Instead, replacing inputs
and outputs on its support gives a free channel that dominates
the corrected branch at a cost exponential only in that support.
Convexity then combines the separate corrections.

\begin{lemma}\label{lem:multiplier}
Let an admissible family also satisfy \ref{ax:permutation} and
\ref{ax:marginal}, with local input $A$ and output $D$,
$a=\dim A$, $d=\dim D$, faithful input state $\tau$, and
faithful replacer state $\omega$.
Put $t=\lambda_{\min}(\tau)$ and $w=\lambda_{\min}(\omega)$.
Suppose $\M\in\F_n$ and $\M\cpge p\Ad_K$, where $p>0$
and $K:A^n\to D^n$.
Let $O=\sum_i c_iF_i$ be a finite expansion into output
contractions of support sizes $r_i$. Define
\begin{equation}
 h_r=(ad)^r[(at)(dw)]^{-r/2},\qquad
 H=\sum_i|c_i|h_{r_i}.
 \label{eq:multiplier-cost}
\end{equation}
If $H>0$, there is $\M'\in\F_n$ such that
\begin{equation}
 \M'\cpge \frac p{H^2}\Ad_{OK}.
 \label{eq:multiplier-order}
\end{equation}
\end{lemma}
The qualification $H>0$ excludes the zero expansion. Since every
$h_r$ is strictly positive, $H=0$ means that all coefficients
vanish and $O=0$; in that case domination is vacuous and no
division by $H$ is needed.
\begin{proof}
For a subset $S$ of $r$ sites, first replace its input by $\tau^{\otimes r}$,
apply $\M$, discard its output, and insert $\omega^{\otimes r}$.
The resulting channel $\M_S$ is free by
\ref{ax:marginal}, \ref{ax:tensor}, and \ref{ax:permutation}.
For the empty subset use $\M$; for the full subset use the free
replacer. As a map on $n$ sites, this construction is
\[
 \M_S=\R_{\omega,S}\circ\M\circ\R_{\tau,S},
\]
with identities on the remaining sites understood.

On one site,
$\R_\tau\cpge(at)\R_{I/a}$ and
$\R_\omega\cpge(dw)\R_{I/d}$.
Using these CP inequalities and the unitary Weyl expansions of the completely depolarizing
input and output channels on $S$ gives the inequality below.
Here an $s$-dimensional unitary Weyl basis consists of $s^2$
unitaries $W_j$, including the identity, with
$\Tr W_i^{\dagger}W_j=s\delta_{ij}$ and
$s^{-2}\sum_jW_jYW_j^{\dagger}=\Tr(Y)I/s$; $\delta_{ij}$ is one if
$i=j$ and zero otherwise. We use their tensor products on $S$.
Each replaced input site contributes a factor $at/a^2$, and
each output site contributes $dw/d^2$. Thus
\begin{equation}
 \M_S\cpge
 p[(at)(dw)]^r d^{-2r}a^{-2r}
 \sum_{P,Q}\Ad_{PKQ}.
 \label{eq:twirl-order}
\end{equation}
Here $P$ and $Q$ range over output and input Weyl bases on $S$,
respectively, and are extended by the identity.
For an output contraction $F=F_S\otimes I$ supported on $S$, write
$F=\sum_P f_PP$. Weyl orthogonality implies
$\sum_P|f_P|^2=d^{-r}\Tr(F_S^{\dagger}F_S)\le1$.
Keeping $Q=I$ in \eqref{eq:twirl-order} is allowed because all
discarded Choi summands are positive. For any Choi-space vector $\ket v$,
\[
 \begin{aligned}
 |\langle v\mid FK\rangle|^2
 &=\left|\sum_P f_P\langle v\mid PK\rangle\right|^2\\
 &\le\Bigl(\sum_P|f_P|^2\Bigr)
       \sum_P|\langle v\mid PK\rangle|^2.
 \end{aligned}
\]
Thus Cauchy--Schwarz on vectorized operators yields
\[
 \sum_P\ket{PK}\bra{PK}\ge\ket{FK}\bra{FK}.
\]
The Choi criterion for CP order therefore gives
$\M_S\cpge p h_r^{-2}\Ad_{FK}$.
The Weyl expansion is an algebraic decomposition of the
depolarizing map. Freeness of $\M_S$ follows from marginal closure,
tensor closure, and permutation invariance.

Omit zero $c_i$ and let
$u_i=|c_i|h_{r_i}/H$. Choose the corresponding free $\M_i$ above.
For arbitrary vectors $\ket{v_i}$, weighted Cauchy--Schwarz gives
\begin{equation*}
 \left|\sum_i c_iv_i\right\rangle
 \left\langle\sum_i c_iv_i\right|
 \le H^2\sum_i u_i h_{r_i}^{-2}\ket{v_i}\bra{v_i},
\end{equation*}
Indeed, with $a_i=u_ih_{r_i}^{-2}>0$ and any test vector $\ket v$,
\[
 \begin{gathered}
 \left|\sum_i c_i\langle v\mid v_i\rangle\right|^2
 \le\left(\sum_i\frac{|c_i|^2}{a_i}\right)
          \sum_i a_i|\langle v\mid v_i\rangle|^2,\\
 \sum_i\frac{|c_i|^2}{a_i}=H\sum_i|c_i|h_{r_i}=H^2.
 \end{gathered}
\]
This proves the operator inequality, including arbitrary complex
phases in the coefficients.
Apply this to the vectorizations of $F_iK$.
The convex mixture $\M'=\sum_i u_i\M_i$ proves
\eqref{eq:multiplier-order}.
\end{proof}

\subsubsection{Approximation of a product-image projector}
The image projector of a tensor-power isometry is itself a product,
but as a single supported operator it generally acts on every site.
Applying the multiplier estimate directly would therefore incur a
linear logarithmic cost. We instead approximate it by a polynomial
of low degree in the number of sites outside the image. A degree
$\ell$ polynomial has a supported-contraction expansion on at most
$\ell$ sites, even though the polynomial acts on the whole block.

\begin{lemma}\label{lem:image-projector}
Let the local space be $A\otimes C_0$, fix a unit vector
$\ket0\in C_0$, and let
$P_k=I_{A^k}\otimes(\ket0\bra0)^{\otimes k}$.
For $k\ge2$ and $0<\xi\le1/16$, there is an operator $H_{k,\xi}$ such that
\begin{equation}
 \opnorm{H_{k,\xi}}\le1,\qquad
 \opnorm{H_{k,\xi}-P_k}\le\xi.
 \label{eq:image-error}
\end{equation}
It has a contraction expansion of support size at most
$\min\{k,\ell\}$ and coefficient cost at most $15^\ell$, where
$\ell=\lceil\sqrt{k}\ln(2/\xi)\rceil$.
\end{lemma}
\begin{proof}
  If $\dim C_0=1$, then $P_k=I$; take $H_{k,\xi}=I$,
  which satisfies all the claimed bounds.
  Otherwise let $q_i$ be the one-site projection onto $\ket0^\perp$ in $C_0$,
  including the identity on $A$, and put $Z=\sum_{i=1}^kq_i$.
  Its spectrum is contained in $\{0,1,\ldots,k\}$, and $P_k$ is the projection
  onto its zero eigenspace.
  Let $T_j$ be the Chebyshev polynomial, with $T_0=1$, $T_1=x$, and
  $T_{j+1}=2xT_j-T_{j-1}$.
  Set
  \begin{equation}
    H_{k,\xi}= \frac{T_\ell((k+1-2Z)/(k-1))}
    {T_\ell((k+1)/(k-1))}.\label{eq:chebyshev}
  \end{equation}
  At $Z=0$ this equals one.
  For $1\le Z\le k$ the argument in the numerator is in $[-1,1]$, so its
  absolute value is at most one.
  Moreover $\operatorname{arcosh}((k+1)/(k-1))\ge1/\sqrt{k}$.
  For example, $\cosh(1/\sqrt{k})\le1+1/k\le1+2/(k-1)$ for $k\ge2$, using the
  power series on $[0,1]$.
  Thus the denominator in \eqref{eq:chebyshev} is at least
  $\tfrac12e^{\ell/\sqrt{k}}$, proving \eqref{eq:image-error}.

  The operator $X=(k+1-2Z)/(k-1)$ has a one-site contraction expansion, also
  allowing the identity, of coefficient cost $(3k+1)/(k-1)\le7$.
  Products of contractions are contractions supported on the union of their
  supports.
  The recurrence gives a support-$j$ expansion of $T_j(X)$ of cost at most
  $15^j$: the induction step costs at most $14\cdot15^j+15^{j-1}<15^{j+1}$.
  Dividing by the denominator, which is at least one, does not increase the
  cost.
\end{proof}

\subsection{Channel completion}\label{sec:proof}
We first correct a positive-overlap branch for an isometric target,
then use auxiliary extensions to treat an arbitrary target.
This proves \cref{thm:completion}, whose quantitative consequences
are derived in \cref{sec:quantitative}.
\Cref{app:alternative-limits} also uses the resulting domination
bound to give an alternative proof of the Stein identity.

\subsubsection{Completion for an isometric target}\label{sec:completion}
The core construction corrects a branch with positive overlap with
an isometric target. Its local corrections are controlled by three
estimates: weighted local-support approximation, free multiplier
completion, and a product-image projector
(\cref{lem:local-approximation,lem:multiplier,lem:image-projector}).
Their quantitative forms were established in \cref{sec:local}.

The purpose of the next proposition is to repair a branch that
has nonzero overlap with the target in every input direction. Its
conclusion preserves the original coefficient $p$, apart from a
controlled loss that is subexponential at fixed or inverse-polynomial
accuracy, while replacing that branch by a channel close to the
target. The first three steps of the proof
obtain symmetry, discard sites, and correct the reduced operator.
Step 4 establishes free domination and restores trace preservation;
Step 5 restores the discarded sites. Step 6 handles the regime in
which the local approximation is not used.

\begin{proposition}[Isometric target]\label{prop:isometric-completion}
Let $J:A\to D$ be an isometry, assume $\dim D$ is a multiple of
$a=\dim A$, and let $\F$ be admissible and also satisfy
\ref{ax:permutation} and \ref{ax:marginal}, with input state $\tau$
and output replacer state $\omega$.
Fix $c>0$. There is a constant $K_c$ with the following property.
If
\begin{equation*}
 \begin{gathered}
 \M\in\F_n,\qquad \M\cpge p\Ad_V,\qquad 0<p\le1,\\
 \opnorm V\le1,\qquad \Rea(J_n^{\dagger}V)\ge cI,
 \qquad J_n=J^{\otimes n},
 \end{gathered}
\end{equation*}
then, for every $0<\delta\le1/16$, there are
$\mathcal S_n\in\F_n$ and a CPTP map $\mathcal L_n$ such that
\begin{align*}
 p\,2^{-K_c n^{2/3}\log_2((n+1)/\delta)}\mathcal L_n
 &\cple\mathcal S_n,
 \\
 \dnorm{\mathcal L_n-\Ad_{J_n}}&\le\delta.
\end{align*}
The constant depends only on the local dimensions, $c$,
$\lambda_{\min}(\tau)$, and $\lambda_{\min}(\omega)$.
It is independent of $n$, $\delta$, $p$, $V$, $J$, and the free family.
\end{proposition}
\begin{proof}
\emph{Step 1: Symmetrization and precision parameters.}
For any unit vector $\ket x\in A^n$, the assumptions give
\[
 c\le\Rea\bra x J_n^{\dagger}V\ket x
 \le\|J_n\ket x\|\,\|V\ket x\|\le1.
\]
Thus $0<c\le1$.

Since $\dim D$ is a multiple of $\dim A$, choose $C_0$ with
$\dim C_0=\dim D/\dim A$ and an output unitary
$U:D\to A\otimes C_0$ such that $UJ\ket x=\ket x\otimes\ket0$.
Such a unitary exists by extending orthonormal bases of the two
isometric images. Applying $U$ at every output site transforms
$\F_r$ into $\{\Ad_{U^{\otimes r}}\circ\M:\M\in\F_r\}$ and
$\omega$ into $U\omega U^{\dagger}$. This change preserves all
family axioms, $\tau$, the eigenvalues of $\omega$, and the norm,
overlap, CP-order, and diamond-distance bounds used here.

Keeping the same notation, identify $D^r\simeq A^r\otimes C_0^r$
and write $J_r=J^{\otimes r}$, so that
$J_r\ket x=\ket x\otimes\ket0^{\otimes r}$.
For $\pi\in S_n$, define
\[
 \begin{gathered}
 V_\pi=U_\pi^D V(U_\pi^A)^{\dagger},\qquad
 \overline{V}=\frac1{n!}\sum_\pi V_\pi,\\
 \M_\pi=\Ad_{U_\pi^D}\circ\M\circ\Ad_{(U_\pi^A)^{\dagger}}.
 \end{gathered}
\]
Let $\overline{\M}=n!^{-1}\sum_\pi\M_\pi$. This channel is free by
permutation invariance and convexity. Since $J_n$ intertwines the
input and output permutations, each $V_\pi$ has norm at most one
and $\Rea(J_n^{\dagger}V_\pi)\ge cI$. Their average satisfies both properties.
The required CP order follows from the covariance identity
\[
 \frac1{n!}\sum_\pi\ket{V_\pi}\bra{V_\pi}
       -\ket{\overline{V}}\bra{\overline{V}}
 =\frac1{n!}\sum_\pi
       \ket{V_\pi-\overline{V}}\bra{V_\pi-\overline{V}}\ge0.
\]
Consequently $\overline{\M}\cpge p\Ad_{\overline{V}}$.
This compares Choi operators, not the generally non-Hermitian
operators $V_\pi$ themselves.
Thus $W=J_n^{\dagger}\overline{V}$ is a permutation-invariant contraction with
$\Rea W\ge cI$.

We choose the approximation precision explicitly.
Put
\[
 \begin{aligned}
 \lambda&=c/2,& q_0&=\sqrt{1-3c^2/4},\\
 q_1&=(1+q_0)/2,& M_c&=\frac{\lambda}{1-q_1}.
 \end{aligned}
\]
and fix
\begin{equation}
 \kappa_c=\min\left\{\frac1{16},
             \frac{q_1-q_0}{2\lambda},
             \frac1{16(2M_c+1)}\right\},\qquad
 \xi=\kappa_c\delta .
 \label{eq:internal-precision}
\end{equation}
Let $\ell_n=\ln((n+1)/\xi)$, $m=\lceil n^{2/3}\ell_n\rceil$,
and $k=n-m$.
First treat the active regime where $n$ exceeds the uniform
threshold in \cref{lem:local-approximation} and $m\le n/2$;
increase the threshold if necessary to ensure $k\ge2$.
In the remaining regime the faithful replacer gives an exact
approximation with a cost covered by the same estimate.

\smallskip\noindent\emph{Step 2: Weighted discarding.}
Keep the first $k$ input-output pairs and discard the last $m$.
The reduced channel is
\begin{equation*}
 \M_k(Y)=\Tr_{D_{k+1}\cdots D_n}
                   \overline{\M}(Y\otimes\tau^{\otimes m}),
 \qquad Y\in\mathrm L(A^k).
\end{equation*}
It belongs to $\F_k$ by \eqref{eq:iterated-marginal}.
Applying this preparation and partial trace to the CP inequality
for $\overline{\M}$ gives a reduced CP branch. We next select a single
coherent Kraus combination from that branch.
Diagonalize $\tau^{\otimes m}=\sum_j u_j\ket j\bra j$.
It follows that $\M_k\cpge p\Ad_K$, where
\begin{equation}
 K=\sum_j u_j
 (I_{D^k}\otimes\bra jJ_m^{\dagger})\overline{V}
 (I_{A^k}\otimes\ket j).
 \label{eq:coherent-marginal}
\end{equation}
To verify this claim, choose an orthonormal basis $\{\ket{e_l}\}$ of
$D^m$ containing the orthonormal vectors $J_m\ket j$.
The marginal of $\Ad_{\overline{V}}$ has Kraus operators
\[
 A_{l,j}=\sqrt{u_j}(I_{D^k}\otimes\bra{e_l})\overline{V}
                         (I_{A^k}\otimes\ket j).
\]
For the basis index $l(j)$ with $\ket{e_{l(j)}}=J_m\ket j$, set
$A_j=A_{l(j),j}$. Equation~\eqref{eq:coherent-marginal} reads
$K=\sum_j\sqrt{u_j}A_j$. There are two square-root weights:
one comes from preparing the mixed input and one from the coherent
Kraus combination. Since $\sum_j u_j=1$, for every vector $\ket v$
in the reduced Choi space,
\[
 |\langle v\mid K\rangle|^2
 \le\sum_j|\langle v\mid A_j\rangle|^2
 \le\sum_{l,j}|\langle v\mid A_{l,j}\rangle|^2.
\]
The Choi criterion proves $\M_k\cpge p\Ad_K$.
No measurement outcome $j$ is recorded in the final branch; the
combination is justified by CP order.
Equation~\eqref{eq:coherent-marginal} is a convex average of
contractions, so $\opnorm K\le1$.
Its compression satisfies
\begin{equation}
 X=J_k^{\dagger}K=\mathbb E_{\tau,m}(W),\qquad
 \opnorm X\le1,\qquad \Rea X\ge cI.
 \label{eq:compression}
\end{equation}
For clarity, the compression is
\[
 J_k^{\dagger}K=\sum_j u_j(I_{A^k}\otimes\bra j)
                         W(I_{A^k}\otimes\ket j),
\]
which is exactly the weighted expectation defined in
\cref{lem:local-approximation}. That expectation is unital and
completely positive, hence contractive in operator norm and
order preserving on Hermitian parts. This proves both bounds
in \eqref{eq:compression}.

\smallskip\noindent\emph{Step 3: Local correction of the reduced branch.}
There are two errors to correct. The component
$(I-J_kJ_k^{\dagger})K$ lies outside the target image, while the component
inside that image is $J_kX$, which is distorted by $X$.
The inequality $\Rea X\ge cI$ implies
\[
 c\langle v|v\rangle\le\Rea\bra v X\ket v
        \le\|\ket v\|\,\|X\ket v\|,
\]
so $X$ is invertible and $\|X^{-1}\|_\infty\le c^{-1}$.
An exact correction would be
\[
 O_{\rm exact}=(X^{-1}\otimes I_{C_0^k})J_kJ_k^{\dagger},
 \qquad O_{\rm exact}K=J_k.
\]
The issue is to control its free-domination cost. We approximate
the inverse and the image projector separately by operators with
controlled local expansions.

Apply \cref{lem:local-approximation} with precision $\xi$.
It provides $P$ with $\opnorm{P-X}\le\xi$, support size
$R=O_{a,t}(n^{2/3})$, and coefficient cost
$M_P\le\exp[O_{a,t}(n^{2/3})]$, where $t=\lambda_{\min}(\tau)$.
For $a=1$, use the scalar $P=X$ directly.
The positive real part of $X$ implies
\begin{equation*}
 \begin{aligned}
 (I-\lambda X)^{\dagger}(I-\lambda X)
 &=I-2\lambda\Rea X+\lambda^2X^{\dagger}X\\
 &\le(1-2\lambda c+\lambda^2)I=(1-3c^2/4)I.
 \end{aligned}
\end{equation*}
By \eqref{eq:internal-precision},
$\opnorm{I-\lambda P}\le q_0+\lambda\xi\le q_1<1$.
Choose an integer $L\ge0$ with $q_1^{L+1}\le\xi$ and
$L=O_c(\ln(1/\xi))$. Define
\begin{equation}
 G=\lambda\sum_{j=0}^{L}(I-\lambda P)^j .
 \label{eq:inverse-polynomial}
\end{equation}
Then
\begin{equation}
 \opnorm G\le M_c,\qquad
 \opnorm{GX-I}\le M_c\xi+q_1^{L+1}\le(M_c+1)\xi .
 \label{eq:inverse-error}
\end{equation}
Here $P$ approximates the compressed operator $X$ and is not a
projection. The finite geometric-series identity gives
$GP=I-(I-\lambda P)^{L+1}$, so
\[
 GX-I=G(X-P)-(I-\lambda P)^{L+1}.
\]
Bounding these two summands proves \eqref{eq:inverse-error}.
Expanding \eqref{eq:inverse-polynomial} gives a contraction
expansion of support at most $LR$ and cost at most
$\lambda(L+1)(1+\lambda M_P)^L$.
Both the support bound and the logarithm of this cost are
$O_{a,t,c}(n^{2/3}\ell_n)$, uniformly in the active regime.

Lift $G$ to $\widehat G=G\otimes I_{C_0^k}$ on $D^k$.
Let $P_k=J_kJ_k^{\dagger}$ and choose $H=H_{k,\xi}$ from
\cref{lem:image-projector}.
Its support and logarithmic cost are
$O(\sqrt{k}\ln(2/\xi))=O(n^{2/3}\ell_n)$.
Thus $O=\widehat GH$ also has a contraction expansion with
support and logarithmic cost $O_{a,d,t,c}(n^{2/3}\ell_n)$,
where $d=\dim D$.
The operator $P_k$ is the exact image projection, whereas $H$
is its local polynomial approximation. The identities
$P_kK=J_kX$ and $\widehat GJ_k=J_kG$ give
\[
 OK-J_k=\widehat G(H-P_k)K+J_k(GX-I).
\]
The leakage and distortion errors are therefore bounded separately:
\begin{equation*}
 \begin{aligned}
 \opnorm{OK-J_k}
 &\le \opnorm G\,\opnorm{H-P_k}\opnorm K\\
 &\quad+\opnorm{GX-I}\\
 &\le (2M_c+1)\xi=:\zeta.
 \end{aligned}
\end{equation*}
The error is controlled by the bounded operator norm of $G$;
its expansion cost enters the CP-domination estimate below.

\smallskip\noindent\emph{Step 4: Free domination and trace-preserving completion.}
Put $w=\lambda_{\min}(\omega)$.
Write the expansion of $O$ just obtained as $O=\sum_i c_iF_i$,
let $r_i$ be its support sizes, and set
$H_*=\sum_i|c_i|h_{r_i}$ using \eqref{eq:multiplier-cost}.
Since $\log_2h_r=O_{a,d,t,w}(r)$, the support and coefficient
bounds from Step 3 control this weighted cost.
Since $\zeta\le\delta/16<1$, the corrected branch is nonzero
and $H_*>0$.
By \cref{lem:multiplier} there is a free $\mathcal K_k$ with
$\mathcal K_k\cpge pH_*^{-2}\Ad_{OK}$, where
\[
 \log_2\max\{1,H_*\}
 =O_{a,d,t,c,w}(n^{2/3}\ell_n).
\]
Set $T=OK/(1+\zeta)$, $h_*=\max\{1,H_*\}$, and $p'=p/h_*^2$.
Then $\opnorm T\le1$, $\opnorm{T-J_k}\le2\zeta$,
$0<p'\le1$, and $\mathcal K_k\cpge p'\Ad_T$.
Complete this branch to the channel
\begin{equation*}
 \mathcal L_T(Y)=TYT^{\dagger}
              +\Tr[(I-T^{\dagger}T)Y]\omega^{\otimes k}.
\end{equation*}
The defect satisfies $0\le D_T=I-T^{\dagger}T\le I$, since $T$ is a
contraction. The second term, denoted by $\mathcal C_T$, has Choi
operator $D_T^T\otimes\omega^{\otimes k}\ge0$, so it is CP,
and for every $Y$,
$\Tr\mathcal L_T(Y)=\Tr T^{\dagger}TY+\Tr D_TY=\Tr Y$.
Thus $\mathcal L_T$ is exactly trace preserving, not merely close
to a trace-preserving map. The bound $D_T\le I$ also gives
$D_T^T\otimes\omega^{\otimes k}\le I\otimes\omega^{\otimes k}$,
and hence
$0\cple\mathcal C_T\cple\R_\omega^{\otimes k}$.
Hence the free mixture
$\mathcal S_k=(\mathcal K_k+\R_\omega^{\otimes k})/2$
dominates $(p'/2)\mathcal L_T$. Indeed,
\[
 \mathcal S_k-\frac{p'}2\mathcal L_T
 =\frac12(\mathcal K_k-p'\Ad_T)
       +\frac12(\R_\omega^{\otimes k}-p'\mathcal C_T)
\]
is CP because $p'\le1$.
Moreover,
\begin{align*}
 \dnorm{\mathcal L_T-\Ad_{J_k}}
 &\le 2\opnorm{T-J_k}+\opnorm{I-T^{\dagger}T}\\
 &\le4\opnorm{T-J_k}\le8\zeta\le\delta .
\end{align*}
We used \eqref{eq:single-kraus-continuity},
$\dnorm{\mathcal C_T}=\opnorm{I-T^{\dagger}T}$, and
$I-T^{\dagger}T=J_k^{\dagger}(J_k-T)+(J_k^{\dagger}-T^{\dagger})T$.

\smallskip\noindent\emph{Step 5: Restoring the discarded sites.}
The local replacer satisfies $\R_\omega\cpge(w/a)\Ad_J$:
its Choi operator is $I_A\otimes\omega\ge wI$ and the Choi
vector of $J$ has squared norm $a$.
Pad the discarded sites by setting
$\mathcal S_n=\mathcal S_k\otimes\R_\omega^{\otimes m}$ and
$\mathcal L_n=\mathcal L_T\otimes\Ad_{J_m}$.
Then
\[
 \mathcal S_n\cpge
 p\,2^{-[1+2\log_2h_*+m\log_2(a/w)]}\mathcal L_n .
\]
Tensoring with the channel $\Ad_{J_m}$ preserves the diamond
error, so it stays at most $\delta$. Each term in the exponent is
bounded by a constant times $n^{2/3}\ell_n$. The constant depends
only on $a,d,t,c,w$. Finally, $\xi=\kappa_c\delta$ implies
$\ell_n=O_c(\log_2((n+1)/\delta))$, proving the claimed bound.

\smallskip\noindent\emph{Step 6: The remaining blocklengths and accuracies.}
For the complementary regime take the exact approximation
$\mathcal L_n=\Ad_{J_n}$ and the free
$\mathcal S_n=\R_\omega^{\otimes n}$.
It suffices to pay exponent $n\log_2(a/w)$.
If $n\ge4$ and $m>n/2$, then
$n^{2/3}\ell_n>n/2-1\ge n/4$.
This absorbs the exact exponent into the claimed bound.
The finitely many smaller $n$, including the fixed threshold,
are absorbed by a further uniform increase of $K_c$.
The proof therefore covers every $n$ and every
$0<\delta\le1/16$ with the same constant.
\end{proof}

\begin{remark}[The scale of the completion cost]
The active regime is a case of the proof, not a restriction on
the statement. Its choice of $m$ balances the cost of restoring
discarded sites against the cost of correcting the reduced branch.
Before substituting $m$, the local approximation has support scale
$O(n\sqrt{h/m}+h)$, with $h=O(\log((n+1)/\xi))$.
Ignoring logarithmic factors suggests balancing $m$ against
$n/\sqrt m$, giving $m\asymp n^{2/3}$. This is a heuristic
explanation of the power. The uniform bound proved above uses
$m=\lceil n^{2/3}\ell_n\rceil$, the inverse degree
$O_c(\log(1/\xi))$, and the explicit complementary-regime estimate.
No optimality of the power $2/3$ is asserted.
\end{remark}

\subsubsection{Completion for a general target}

To pass from this construction to an arbitrary target, we use the
lemmas on branch extraction and comparison with auxiliary extensions
(\cref{lem:branch,lem:extension}), established in
\cref{app:testing-lifts}. Branch extraction supplies the positive
overlap from a testing lower bound. The extension comparison makes
the resulting bound applicable to a Stinespring dilation with
constants independent of the blocklength.

\begin{proof}[Proof of \cref{thm:completion}]
Fix a desired accuracy $0<\nu\le1/16$.
Choose an auxiliary space $E$ of dimension $ab$ and a Stinespring
isometry $V:A\to B\otimes E$ for $\N$; the Kraus-rank bound
allows this choice \cite[Thm.~2.22]{Watrous}.
The local lifted output dimension is $ab^2$, a multiple of $a$.
For each $n$ define
\begin{equation}
 \mathcal G_n=
 \{\mathcal T\in\CPTP(A^n\to(B\otimes E)^n):
                \Tr_{E^n}\circ\mathcal T\in\F_n\}.
 \label{eq:lifted-family}
\end{equation}
We check every axiom.
The family is nonempty because any $\M\in\F_n$ can be extended
by a fixed auxiliary state.
Taking the $B$ marginal is a continuous linear map on Choi
operators. Its inverse image of the closed convex set $\F_n$,
intersected with the compact set of channels, is therefore compact
and convex.
Its permutation and tensor-product closures follow by taking
$B$ marginals.
It contains the tensor powers of the replacer channel with faithful
output state $\omega\otimes I_E/(ab)$.
Finally, insert the same fixed $\tau$ at the last input and
discard the corresponding $BE$ output.
The $B$ marginal of the remaining channel equals the
$\tau$-insertion marginal of an element of $\F_n$, so it lies in
$\F_{n-1}$ by \ref{ax:marginal}.
Thus the lifted family satisfies the one-site marginal axiom with
exactly the original input state $\tau$.

Let $\gamma_{\delta,n}$ be the worst-case testing value for target
$\Ad_{V^{\otimes n}}$ against $\mathcal G_n$, at type-I error
at most $\delta$. This $\delta$ is a testing tolerance; $\nu$
remains the desired diamond accuracy.
Apply \cref{lem:extension} to each block of $n$ uses, choosing
$\delta,t$ from the fixed $\eps$. Then
\begin{equation}
 \gamma_{\delta,n}\ge t\beta_{\eps,n}.
 \label{eq:lifted-beta}
\end{equation}
These constants do not depend on $n$.
By \cref{lem:testing-minimax}, a hardest lifted alternative
$\widetilde\M_n\in\mathcal G_n$ exists with
$b_\delta(\Ad_{V^{\otimes n}},\widetilde\M_n)
=\gamma_{\delta,n}$.
The faithful replacer and \cref{lem:bounds} imply
$\gamma_{\delta,n}>0$.
By \cref{lem:branch}, this alternative CP-dominates
$\gamma_{\delta,n}\Ad_{W_n}$ for an operator satisfying
$\opnorm{W_n}\le1$ and
$\Rea[(V^{\otimes n})^{\dagger}W_n]\ge c_\delta I$, where
$c_\delta=(1-\sqrt{1-\delta})/(1+\sqrt{1-\delta})>0$
is independent of $n$.

Apply \cref{prop:isometric-completion} in the lifted family at
accuracy $\nu$.
It supplies a free $\widetilde{\mathcal S}_n$ and a channel
$\widetilde{\mathcal L}_n$ with
\[
 \widetilde{\mathcal S}_n\cpge
 \gamma_{\delta,n}\,2^{-h_\delta(n,\nu)}
 \widetilde{\mathcal L}_n,\qquad
 \dnorm{\widetilde{\mathcal L}_n-\Ad_{V^{\otimes n}}}
 \le\nu,
\]
where $h_\delta(n,\nu)=O_\delta(n^{2/3}\log_2((n+1)/\nu))$,
with a constant uniform in $n$ and $\nu$.
Trace out $E^n$.
The resulting $\mathcal S_n$ belongs to $\F_n$ by
\eqref{eq:lifted-family}, the resulting $\mathcal L_n$ is CPTP,
and postcomposition with the CPTP partial-trace map preserves
CP order and cannot increase the diamond distance.
Use \eqref{eq:lifted-beta} and absorb $\log_2(1/t)$ into
$K_\eps n^{2/3}\log_2((n+1)/\nu)$ for all $n\ge1$.
This gives the required CP domination and diamond-norm
approximation on the original output space.
\end{proof}

\section{Symmetric entropy estimates and an alternative limit proof}
\label{app:alternative-proof}

The entropy minimax identity below holds under the basic assumptions.
Permutation invariance gives a further finite-block entropy estimate.
With marginal closure as well, quantitative completion and near
subadditivity yield an alternative proof of \cref{thm:main}.
This proof supplies a second route to the ordinary limits through
the explicit completion bound of \cref{thm:completion}.

\subsection{Entropy minimax and a large-error bound}\label{app:symmetric-entropy}
We next use concavity in the input marginal to exchange
the input and alternative optimizations. For $\rho\in\Dens(A^n)$,
let $\ket{\chi_\rho}=(I\otimes\sqrt\rho)\ket{\Gamma_{A^n}}$.
For a channel $\M:\mathrm L(A^n)\to\mathrm L(B^n)$, define
\begin{equation*}
 f(\rho,\M)=D\bigl((\id\otimes\N_n)(\ket{\chi_\rho}\bra{\chi_\rho})
               \Vert(\id\otimes\M)(\ket{\chi_\rho}\bra{\chi_\rho})\bigr).
\end{equation*}
The state $\ket{\chi_\rho}$ is the canonical purification of the physical input marginal
$\rho$. In the vectorization convention of \eqref{eq:choi},
$\ket{\chi_\rho}=\ket{\sqrt\rho}$ and its reference marginal
is $\rho^T$, consistently with the tester convention in
\eqref{eq:tester}.

\begin{lemma}\label{lem:entropy-minimax}
The function $f$ is concave in $\rho$ and convex in $\M$, allowing
extended values. Moreover,
\begin{equation}
 E_n=\max_{\rho\in\Dens(A^n)}
       \inf_{\M\in\F_n}f(\rho,\M).
 \label{eq:entropy-minimax}
\end{equation}
Here permutation covariance means
$\M\circ\Ad_{U_\pi^A}=\Ad_{U_\pi^B}\circ\M$ for every $\pi$.
If $\M$ is permutation-covariant and has a positive replacer
component, a permutation-invariant $\rho$ attains
$D_{\ch}(\N_n\Vert\M)$.
Having a positive replacer component means that
$\M\cpge\eta\R_\omega^{\otimes n}$ for some $\eta>0$.
\end{lemma}
\begin{proof}
\emph{Step 1: Concavity and convexity.}
Purification and Schmidt compression give
$D_{\ch}(\N_n\Vert\M)=\sup_\rho f(\rho,\M)$.
For $\rho=\sum_i p_i\rho_i$, the vector
$\sum_i\sqrt{p_i}\ket i\ket{\chi_{\rho_i}}$ is another purification
of $\rho$. All purifications of a fixed input marginal are related by
isometries on a sufficiently enlarged reference. The flagged vector
may therefore be used to compute $f(\rho,\M)$. Dephasing its flag
produces a direct sum of the output pairs for $\rho_i$, with common
weights $p_i$. Dephasing is a reference channel. Data processing
and the direct-sum identity therefore give
$f(\rho,\M)\ge\sum_i p_i f(\rho_i,\M)$.
Convexity in $\M$ follows from joint convexity of relative entropy.

\emph{Step 2: Faithful regularization.}
For $0<\delta<1$ put
$\M_\delta=(1-\delta)\M+\delta\R_\omega^{\otimes n}$ and
$f_\delta(\rho,\M)=f(\rho,\M_\delta)$.
Denote the target output on $\ket{\chi_\rho}$ by $\varrho$ and
the $\M_\delta$ output by $\sigma$.
These states obey $\varrho\le K_\delta \sigma$ with
$K_\delta=2^{nC}/\delta$, by \eqref{eq:replacer-domination}.
This makes $f_\delta$ finite and jointly continuous, including at
singular input marginals. Here is a useful direct verification.
On pairs of states of a fixed dimension $D_0$ satisfying $\varrho\le K \sigma$,
replace $\log_2 \sigma$ by $\log_{2,t}\sigma$, defined by replacing each
eigenvalue $\lambda$ of $\sigma$ by $\max\{\lambda,t\}$ before taking
the logarithm, where $t>0$. Write $\lambda_j(\sigma)$ for the
eigenvalues of $\sigma$. Then
\begin{align*}
 0&\le\Tr \varrho(\log_{2,t}\sigma-\log_2 \sigma)\\
 &\le K\sum_{0<\lambda_j(\sigma)<t}
       \lambda_j(\sigma)\log_2\frac t{\lambda_j(\sigma)}
 \le\frac{KD_0t}{e\ln2}.
\end{align*}
Zero eigenvalues contribute nothing because $\varrho\le K\sigma$.
Thus the cross-entropy is a uniform limit of continuous functions;
$\Tr \varrho\log_2 \varrho$ is continuous as well. Canonical purification is
continuous in $\rho$.

\emph{Step 3: Passage to the unregularized problem.}
Minimax now applies to $f_\delta$. In the rest of this argument,
every infimum over $\M$ is over $\F_n$. If
$g(\rho)=\inf_\M f(\rho,\M)$ and
$g_\delta(\rho)=\inf_\M f_\delta(\rho,\M)$, then
\begin{equation}
 0\le g_\delta(\rho)-g(\rho)\le-\log_2(1-\delta).
 \label{eq:regularization-error}
\end{equation}
The lower bound holds because all $\M_\delta$ belong to $\F_n$;
the upper bound follows from
$\M_\delta\cpge(1-\delta)\M$ and operator monotonicity of logarithm.
The same bounds hold between
$\inf_\M\max_\rho f_\delta(\rho,\M)$ and $E_n$.
Letting $\delta\downarrow0$ proves \eqref{eq:entropy-minimax}.
Each $g_\delta$ is continuous by compactness and joint continuity;
\eqref{eq:regularization-error} implies uniform convergence to $g$.
Consequently the maximum is attained.

\emph{Step 4: Symmetry.}
For a fixed permutation-covariant $\M$ with a positive replacer
component, $f(\cdot,\M)$ is continuous by the same argument.
Both $\M$ and the tensor-power target $\N_n$ are
permutation-covariant. Conjugating the input marginal by $U_\pi^A$
therefore conjugates both output states by the same reference-output
unitary. Unitary invariance of relative entropy gives
$f(U_\pi^A\rho(U_\pi^A)^{\dagger},\M)=f(\rho,\M)$.
Its concavity shows that averaging a maximizing $\rho$ over
permutations preserves maximality.
\end{proof}

Permutation symmetry now gives a finite-block lower bound on
the testing exponent in terms of the optimized entropy.

\begin{lemma}\label{lem:large-error}
Let $\F$ be admissible and also satisfy \ref{ax:permutation}.
Set $h=ab$ and
$v_n=\binom{n+h^2-1}{h^2-1}$. For $0<\eps<1$,
\begin{equation}
 -\log_2\beta_{\eps,n}
 \ge\frac{E_n-(1-\eps)(nC+1)-\log_2 v_n}{\eps}-1 .
 \label{eq:large-error}
\end{equation}
\end{lemma}
\begin{proof}
\emph{Step 1: A symmetric hardest alternative.}
As a function of $\M$, $b_\eps(\N_n,\M)$ is concave, since it
is the infimum of linear acceptance probabilities over a fixed
feasible tester set. Permuting the input and output of a tester
preserves its target acceptance because $\N_n$ is permutation-covariant.
This bijection of feasible testers makes $b_\eps(\N_n,\M)$
invariant under simultaneous input-output permutations of $\M$.
Average a hardest alternative from \cref{lem:testing-minimax}
over these permutations. Convexity and permutation closure keep
the average in $\F_n$, while concavity preserves maximality.
The resulting hardest alternative is permutation-covariant.

Fix any such $\M$ and set
$\M'=(\M+\R_\omega^{\otimes n})/2$.
Choose a permutation-invariant maximizing marginal for
$D_{\ch}(\N_n\Vert\M')$, using \cref{lem:entropy-minimax}.
Let $\rho,\sigma$ be the target and $\M'$ outputs on its canonical purification.
Then $D(\rho\Vert \sigma)\ge E_n$ and $\rho\le2^{nC+1}\sigma$.
Both outputs are invariant under pair permutations on
$(\mathbb C^{ab})^{\otimes n}$. To see this, the maximizing input
marginal and its square root commute with every input permutation.
Permutation matrices are real in the product basis, so the
canonical purification is fixed by the same permutation on its
reference and input factors. Covariance of both $\N_n$ and $\M'$
then transfers this invariance to their reference-output states.

\emph{Step 2: Polynomial spectral complexity.}
Every permutation-invariant operator on $(\mathbb C^h)^{\otimes n}$
lies in an algebra of dimension at most $v_n$: expand in tensor
products of the $h^2$ matrix units and group terms into permutation
orbits. Each orbit is determined by $h^2$ multiplicities summing to $n$.
A Hermitian element of this algebra has at most $v_n$ distinct
eigenvalues, since its spectral projectors, equivalently suitable
polynomials in that element, are linearly independent.

\emph{Step 3: Reduction to a classical likelihood ratio.}
Write the distinct-eigenvalue decomposition as
$\sigma=\sum_{j=1}^{v}\mu_jP_j$, with $v\le v_n$, and define
\[
 \mathcal P_\sigma(Z)=\sum_{j=1}^{v}P_jZP_j,\qquad
 \rho'=\mathcal P_\sigma(\rho).
\]
Then $\rho'$ commutes with $\sigma$, and pinching preserves traces
against every operator commuting with the $P_j$.
Pinching into $v$ orthogonal blocks is an average of $v$ unitary
conjugations, so its entropy increase is at most $\log_2 v$.
Indeed, the entropy of a mixture is at most the average entropy
plus the Shannon entropy of its mixing weights, while each unitary
conjugate has entropy $S(\rho)=-\Tr \rho\log_2\rho$.
Moreover $\Tr \rho'\log_2\sigma=\Tr \rho\log_2\sigma$, so
$D(\rho\Vert \sigma)-D(\rho'\Vert \sigma)=S(\rho')-S(\rho)\le\log_2v$.
Since pinching fixes $\sigma$,
\begin{equation*}
 D(\rho'\Vert \sigma)\ge D(\rho\Vert \sigma)-\log_2 v_n
 \ge E_n-\log_2 v_n. 
\end{equation*}
Also $\rho'\le2^L \sigma$, where $L=nC+1$.
Diagonalize the commuting states $\rho',\sigma$, writing their
diagonal entries as $\rho'_j$ and $\sigma_j$, respectively.
Under the probability distribution given by $\rho'$, the
log-likelihood variable $Z=\log_2(\rho'_j/\sigma_j)$ is bounded above by $L$
and has expectation at least $E_n-\log_2v_n$.
Zero-probability coordinates are omitted.

\emph{Step 4: A feasible threshold test.}
For $t<L$, write $p=\Pr(Z\ge t)$. The bound
$\mathbb E Z\le t+(L-t)p$ yields
\begin{equation}
 p\ge\frac{E_n-\log_2v_n-t}{L-t}.
 \label{eq:likelihood-threshold}
\end{equation}
Choose a common eigenbasis of $\rho'$ and $\sigma$ and let $Q_t$ project
onto the coordinates with $\rho'_j/\sigma_j\ge2^t$. Coordinates outside
the support of $\rho'$ may be excluded. Then
\[
 \begin{aligned}
 \Tr Q_t\rho&=\Tr Q_t\rho'=p,\\
 \Tr Q_t\sigma&=\sum_{Z_j\ge t}\sigma_j
       \le2^{-t}\sum_{Z_j\ge t}\rho'_j\le2^{-t}.
 \end{aligned}
\]
The first equality follows because $Q_t$ is fixed by the pinching. The output $\sigma_0$ of the original $\M$ satisfies
$\sigma_0\le2\sigma$, so its acceptance is at most $2^{1-t}$.
Choose
\[
 t=\frac{E_n-\log_2v_n-(1-\eps)L}{\eps}.
\]
One has $t<L$ by \cref{lem:bounds}, and substituting this $t$ in
\eqref{eq:likelihood-threshold} gives $p\ge1-\eps$.
The resulting bound on $b_\eps(\N_n,\M)$ is uniform in the
permutation-covariant alternative. Minimax and the hardest-alternative
symmetry prove \eqref{eq:large-error}.
\end{proof}

\subsection{Convergence by near subadditivity}\label{app:alternative-limits}
Throughout this subsection, assume \cref{def:axioms} and both
\ref{ax:permutation} and \ref{ax:marginal}.
The quantitative completion and error-comparison results of
\cref{sec:quantitative}, together with \cref{lem:large-error},
give the following alternative limit argument.

The next elementary sequence estimate converts the completion
overhead into convergence of the normalized testing quantity.
Sublinearity makes a single block's loss negligible per use, but
the concatenation argument must also control the accumulated loss
over many scales. What we use is the summability of the dyadic
errors below. An arbitrary $o(n)$ remainder is not a substitute
for this estimate.

\begin{lemma}\label{lem:near-subadditive}
Suppose $0\le a_n\le C_0n+C_1$ and
\begin{equation}
 a_{n+m}\le a_n+a_m+F(n+m)
 \label{eq:near-subadditivity}
\end{equation}
for all positive $n,m$, where $F$ is nonnegative, nondecreasing, and
$F(s)=O(s^{2/3}\log_2(s+1))$.
Then $a_n/n$ has a finite ordinary limit.
\end{lemma}
\begin{proof}
Fix $k\ge1$. Merge $m$ blocks of length $k$ in balanced rounds,
carrying an unpaired block to the next round.
At round $j$, each merged block has size at most $2^jk$.
For every nonempty round, there are at most $2m/2^j$ merges.
Iterating \eqref{eq:near-subadditivity} therefore gives
\begin{equation*}
 \frac{a_{mk}}{mk}
 \le\frac{a_k}{k}
 +2\sum_{j\ge1}\frac{F(2^jk)}{2^jk}.
\end{equation*}
The infinite sum is an upper bound on the finite list of merge
costs, all of which are nonnegative.
Its terms are bounded by a constant times
$k^{-1/3}2^{-j/3}[\log_2(k+1)+j]$.
The series converges, and its sum is
$O(k^{-1/3}\log_2(k+1))$, hence tends to zero as $k\to\infty$.

For general $n=mk+r$, $0\le r<k$, an additional merge when $r>0$
yields $a_n\le a_{mk}+a_r+F(n)$.
For fixed $k$ the remainder contribution divided by $n$ vanishes.
Since $mk/n\to1$ and $a_{mk}/(mk)$ is uniformly bounded, we obtain
\[
 \limsup_n\frac{a_n}{n}
 \le\frac{a_k}{k}
 +2\sum_{j\ge1}\frac{F(2^jk)}{2^jk}.
\]
Choose a sequence of integers $k$ along which $a_k/k$ tends to
its lower limit. The series on the right tends to zero along
this sequence. The upper limit of $a_n/n$ is therefore no larger
than its lower limit, which proves convergence.
\end{proof}

\begin{proof}[Alternative proof of \cref{thm:main} under all five conditions]
\emph{Step 1: Existence of the operational limit.}
Set $\eps_0=1/2$ and let $a_n=-\log_2\beta_{\eps_0,n}$.
The uniform bounds imply $0\le a_n\le nC+1$.
Set $\eta_n=\min\{1/16,n^{-2}\}$. Since
$\log_2((n+1)/\eta_n)=O(\log_2(n+1))$, applying
\cref{thm:completion} at tolerance $\eps_0$ and accuracy $\eta_n$
gives channels $\mathcal S_n\in\F_n$ and $\mathcal L_n$ satisfying
\[
 \begin{gathered}
 \mathcal S_n\cpge\beta_{\eps_0,n}2^{-f(n)}\mathcal L_n,
 \qquad \dnorm{\mathcal L_n-\N_n}\le\eta_n,\\
 f(n)=K'_{\eps_0}n^{2/3}\log_2(n+1)=o(n).
 \end{gathered}
\]
Tensor closure and CP order give
\[
 \mathcal S_n\otimes\mathcal S_m\cpge
 \beta_{\eps_0,n}\beta_{\eps_0,m}
 2^{-f(n)-f(m)}(\mathcal L_n\otimes\mathcal L_m).
\]
The identity
\[
 \mathcal L_n\otimes\mathcal L_m-\N_n\otimes\N_m
 =(\mathcal L_n-\N_n)\otimes\mathcal L_m
   +\N_n\otimes(\mathcal L_m-\N_m)
\]
and the fact that $\mathcal L_m$ and $\N_n$ are channels give
diamond error at most
$\eta_n+\eta_m\le1/8$ from $\N_{n+m}$.
Thus every $(n+m)$-use tester of target acceptance at least $1/2$
accepts $\mathcal L_n\otimes\mathcal L_m$ with probability
at least $1/4$. This statement includes inputs entangled across
the two blocks, because the bound is in diamond norm.
It follows that
\begin{equation*}
 \begin{aligned}
 a_{n+m}&\le a_n+a_m+f(n)+f(m)+2\\
        &\le a_n+a_m+F(n+m),
 \end{aligned}
\end{equation*}
where $F(s)=2f(s)+2$.
The hypotheses of \cref{lem:near-subadditive} hold. Hence $a_n/n$
has a finite limit $R$, and \cref{lem:error-independence} gives
\begin{equation}
 -\frac1n\log_2\beta_{\eps,n}\longrightarrow R
 \quad\text{for every fixed }0<\eps<1.
 \label{eq:operational-limit}
\end{equation}

\emph{Step 2: Identification of the entropy limit.}
The finite-block entropy bounds now identify the common limit.
The weak converse \eqref{eq:weak-converse} gives
$\liminf_n E_n/n\ge(1-\eps)R$ for every fixed $\eps$.
Letting $\eps\downarrow0$ after this limit yields
$\liminf_n E_n/n\ge R$.
On the other hand, \cref{lem:large-error} rearranges to
\begin{equation*}
 E_n\le \eps(-\log_2\beta_{\eps,n})
 +(1-\eps)(nC+1)+\log_2v_n+\eps.
\end{equation*}
Since $\log_2v_n=O(\log n)$, \eqref{eq:operational-limit}
implies $\limsup_n E_n/n\le\eps R+(1-\eps)C$.
Now let $\eps\uparrow1$ to obtain $\limsup_n E_n/n\le R$.
Thus the ordinary entropy limit exists and equals $R$.
The bound $0\le R\le C$ follows from \cref{lem:bounds}.
\end{proof}

\section*{Acknowledgment}
OpenAI Codex was used in proof development, mathematical checking,
literature work, and preparation of the source manuscript, and in
expanding, reorganizing, and typesetting the present exposition.


\end{document}